\documentclass[11pt]{article}
\usepackage[margin=1in]{geometry}
\usepackage{graphicx}
\usepackage[backend=biber,style=numeric]{biblatex}
\usepackage{amsmath}
\usepackage{mathtools}
\usepackage{amssymb}
\usepackage{amsthm}
\usepackage{authblk}
\usepackage{xurl}

\theoremstyle{definition}
\newtheorem{example}{Example}

\newcount\Comments
\usepackage{color}
\usepackage{xcolor}
\usepackage{csquotes}
\usepackage{thmtools,thm-restate}
\usepackage{cleveref}

\definecolor{darkgreen}{rgb}{0,0.5,0}

\newcommand{\kate}[1]{}
\newcommand{\charlotte}[1]{}
\newcommand{\manish}[1]{}

\newcommand{\argmax}{\mathop{\arg\max}\limits}
\DeclarePairedDelimiter{\parens}{(}{)}

\newcommand{\eu}[1]{\mathbb{E}[U(#1)]}

\newtheorem{definition}{Definition}

\title{\textbf{When to Ask a Question: Understanding Communication Strategies in Generative AI Tools}}

\author{Charlotte Park}
\author{Kate Donahue}
\author{Manish Raghavan}

\affil{Massachusetts Institute of Technology\\
Cambridge, Massachusetts, USA}
\affil{\texttt{\{cispark, kpd46, mragh\}@mit.edu}}
\date{}

\begin{document}

\maketitle

\begin{abstract}
\label{sec-abstract}

Generative AI models differ from traditional machine learning tools in that they allow users to provide as much or as little information as they choose in their inputs. This flexibility often leads users to omit certain details, relying on the models to infer and fill in under-specified information based on distributional knowledge of user preferences. Such inferences may privilege majority viewpoints and disadvantage users with atypical preferences, raising concerns about fairness.
Unlike more traditional recommender systems, LLMs can explicitly solicit more information from users through natural language.
However, while directly eliciting user preferences could increase personalization and mitigate inequality, excessive querying places a burden on users who value efficiency.
We develop a stylized model of user-LLM interaction and develop an objective that captures tradeoff between user burden and preference representation.
Building on the observation that individual preferences are often correlated, we analyze how AI systems should balance inference and elicitation, characterizing the optimal amount of information to solicit before content generation.
Ultimately, we show that information elicitation can mitigate the systematic biases of preference inference, enabling the design of generative tools that better incorporate diverse user perspectives while maintaining efficiency. We complement this theoretical analysis with an empirical evaluation illustrating the model’s predictions and exploring their practical implications.

\end{abstract}

\section{Introduction}
\label{sec-introduction}
\newcommand{\mr}[1]{\textcolor{olive}{[MR: #1]}}
Generative AI models present a new frontier for personalization. Across a wide range of applications and modalities, users with heterogeneous preferences interact with generative models in the hopes of producing content they find desirable. Compared to more traditional recommender systems, generative AI models offer new opportunities: instead of having to learn from user behavior, conversational interfaces enable direct user feedback.
A user might, for example, instruct a tool to add more color to an image or make an email less formal.
On the other hand, their flexibility introduces new challenges as well. Because they support such a broad set of downstream use cases, generative AI models perpetually face ``cold-start’’ problems. Based on limited relevant information about a user’s preferences, a model must produce content that the user finds desirable.\footnote{In recognition of this lack of data, providers of AI chatbots have begun to save user data or connect to external sources to improve personalization.}

Consider an image-generation model tasked with the prompt: ``Make a photorealistic image of a robot relaxing under a tree.'' While the prompt encodes the broad strokes of what the user wants, it leaves far more questions unanswered (What kind of tree? What color is the robot? Is it reading a book?). In the absence of concrete information from the user, the model must fill in the blanks by (perhaps implicitly) making inferences about the user’s ideal image.

In a sense, these implicit inferences about preferences offer convenience to users.
A user may hold weak preferences about the type of tree in the image or the color of the robot, but it may not be worth their time to communicate each of their weak preferences to a model.
Viewed this way, generative models serve as (noisy and biased) decompression algorithms, taking in a compressed description of the user’s desired end product and inflating it by making guesses about all of the attributes that the user left unspecified~\cite{kreminski2025endless}.
Despite the fact that they may not produce a user’s ideal output, they provide utility by saving users time.
A user might be able to, say, write an email without any AI assistance that perfectly aligns with their preferences, but they would be willing to accept a slightly worse-aligned email if a model could generate one with minimal supervision.
In an idealized setting, a user would communicate precisely the preferences for which the marginal time required to do so is outweighed by the marginal utility gain from having those preferences satisfied.

Viewed this way, a rational user can simultaneously benefit from the use of generative models and still end up with suboptimal content.
The user does not get exactly what they want, but AI assistance allows them to produce something ``good enough'' with dramatically less effort than they would otherwise need.
While the user could continue iterating towards their ideal content, it is simply not worth their time to do so.

\paragraph*{Motivating examples.}
Throughout this paper, we will use several motivating examples to illustrate settings where such user-LLM interactions may occur. 
\begin{example}[Vacation planning]\label{ex:vacation}
A user is planning a trip to a city, and queries an LLM to ask about a suggested
itinerary. Based on its training data, the LLM ``knows'' that visitors to
this city tend to fall in a few rough patterns, either preferring to visit
locations $\{A, B, C, \ldots\}$ (e.g.,~museums and historical sites) or a
different set of locations $\{X, Y, Z, \ldots\}$ (e.g.,~restaurants and cafes).
Because a user doesn't know the city, she doesn't know her general
``type'' beforehand. An LLM can learn more about her relative preferences
by asking direct queries (``Would you prefer location $A$ or $X$?'') and
then inferring the rest of her preferences.
\end{example}
\begin{example}[Email writing]\label{ex:email}
 A user is relying on an LLM to help her write an email. While each email is unique, they do tend to fall into specific patterns, such as ``long, polite'' emails or ``short, direct'' emails. While users may not respond well to being asked ``Is this a polite email or not?'', they may be willing to respond to proxies, such as ``Should this email begin with a salutation or directly jump into subject matter?''.
\end{example}
\begin{example}[Image generation]\label{ex:image}
Consider a user who is relying on an LLM to help generate an image. Image generation is a very high-dimensional problem, but one with strong degrees of correlation (e.g. pixel-to-pixel correlation in output, but also common themes in image styles). While a user may not directly know or be able to articulate her preferences, she may be able to answer queries about style aspects or examples that she wishes to include. 
\end{example}
In all of these examples, user preferences may be idiosyncratic, or may be tied to sensitive features, such as socioeconomic status, race, or gender. We seek to model how preference underspecification impacts user utility, with a particular focus on implications for fairness.

\paragraph*{System-level effects of personalization under uncertainty.}

We consider an idealized AI system that, given partial information about the user's preferences, makes inferences based on its knowledge of the population distribution.
A model might infer, for example, that ``a flag next to the Statue of Liberty'' refers to the American flag, while ``a flag next to the Eiffel Tower'' refers to the French flag.

When left to ``fill in the blanks,'' models can have subtle but profound implications on the utility, fairness, and overall diversity of a system.
Inferred preferences may skew towards majority populations and tastes~\cite{chen2018my, zhu2021fairness, kleinberg2019simplicity, wu2022big, li2023transferable} 
As a result, outputs will become homogeneous precisely in those preference dimensions left unspecified.
Thus, the underspecification of preferences serves as a microfoundation for algorithmic homogenization \cite{kleinberg2021algorithmic, bommasani2022picking,castro2023human}.

\paragraph*{The impacts of user interaction.}
The challenges described---incomplete information, user fairness, and content homogeneity---are not unique to generative AI models. Similar concerns arise in more traditional recommender systems when dealing with limited user information. While this is true to some extent, a key differentiator of generative AI lies in its inherent interactivity. Unlike many traditional recommendation systems where user feedback is often implicit (e.g., clicks, purchases), content generation can involve multiple rounds of explicit interaction between the user and the model. A user can iteratively refine the generated output through natural language commands or other interactive modalities, providing direct feedback and steering the model towards their desired outcome.
Moreover, a model can \textit{actively solicit} the user's preferences by asking additional questions.
In fact, many state-of-the-art chatbots now exhibit this behavior (see \Cref{fig:chatgpt}).

\begin{figure}[h!]
    \centering
    \includegraphics[width=0.75\linewidth]{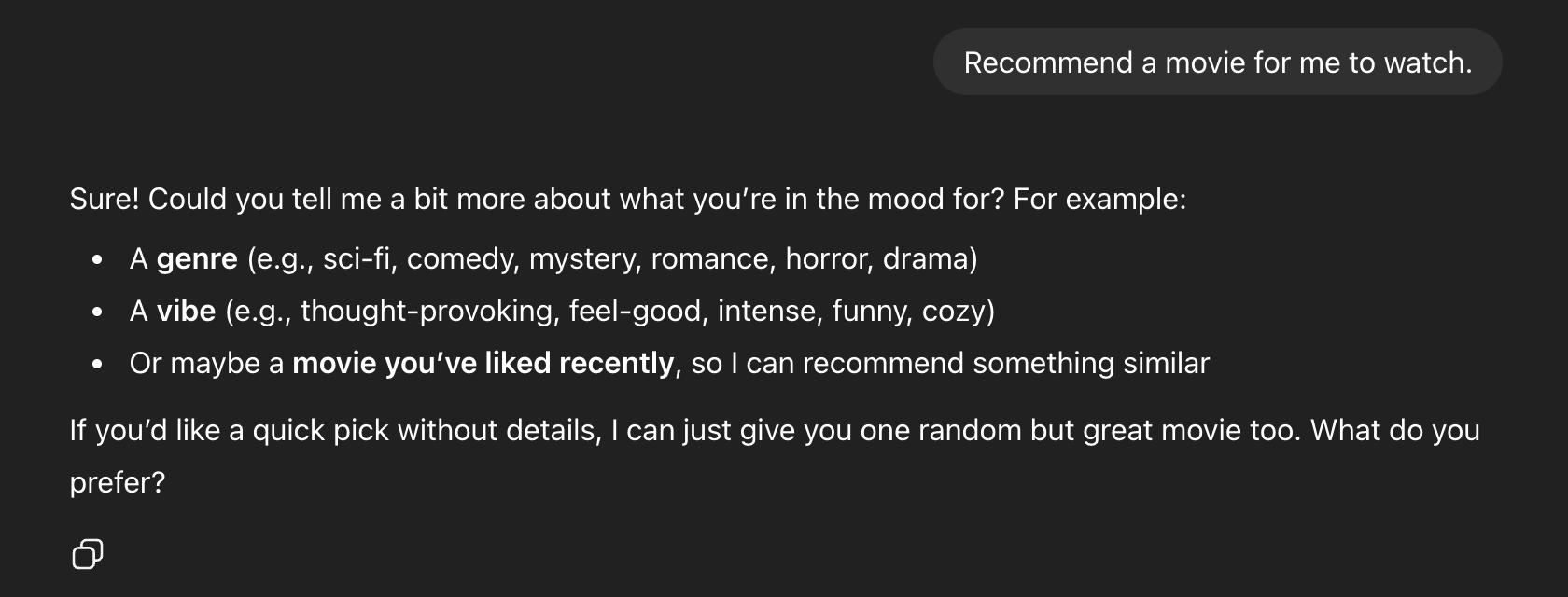}
    \caption{A screenshot of ChatGPT prompting the user for additional information when an under-specified prompt is provided by the user. A decision to solicit information trades off additional user effort (to respond to the LLM's question) for the benefits of better personalization.}
    \label{fig:chatgpt}
\end{figure}

\paragraph*{The present work: Modeling interactive content generation.}

We seek to understand how inferences about preferences impact the utility, fairness, and content diversity of AI systems.
In our model, an AI system has two levers at its disposal.
It can (1) solicit more information about a user's preferences and (2) generate content by making inferences about underspecified preferences.
At a high level, we might expect that soliciting information should increase both fairness and output diversity simply by more faithfully representing user preferences.
While this is true, additional information is not free; providing more information is costly for users (otherwise they would specify their complete preferences up-front).
Thus, we build a formal framework under which we can pose and answer questions about when information-solicitation is ``worth'' the cost.

Our paper is organized as follows.
After surveying related work in \Cref{sec-related-work}, we present our formal model in \Cref{sec-model}.
In \Cref{sec-querying} we study this problem through the lens of maximizing the user's expected welfare, in \Cref{sec-fairness} we turn to broader notions of user welfare (e.g. relating to fairness), and in \Cref{sec-diversity} we briefly discuss implications for output homogenization.
In \Cref{sec-empirics}, we validate the qualitative predictions of our model
with an empirical exercise.
We conclude in \Cref{sec-discussion} with a discussion of our results and future work. 

\section{Related Work}
\label{sec-related-work}
There has been a rich body of literature, both empirical and theoretical, on preference elicitation and social welfare in content creation.  

First, the emerging area of pluralistic alignment (e.g. \cite{sorensen2024roadmap}) aims to align LLMs with a range of human preferences over outcomes (see related summary papers \cite{xie2025survey, guan2025survey, kirk2023personalisation}). Most closely related to our work are papers that connect with voting and social choice literature to aggregate diverse user preferences (e.g. \cite{conitzer2024social, dai2024mapping, shirali2025direct, a2024policy, ge2024axioms, pardeshi2024learning, chen2024pal, freedman2026adaptive, srewa2026appa}). One key difference between most of this literature and our work is that it tends to focus on a one-shot setting where the goal of an algorithm is to represent multiple diverse viewpoints, whereas our focus is when the algorithm has the ability to interact with the human (through querying or critique of content) to better suit their preferences. 

Another related area of literature is that of algorithmic monoculture or homogenization (e.g. \cite{kleinberg2021algorithmic, bommasani2022picking}), which studies the impact of a ``monoculture'' of algorithmic tools on societal welfare. Our paper relates to this body of work by demonstrating how an attempt to satisfy multiple diverse users with only partial information about preferences can lead to homogenization of outcomes. A closely related area of research studies human-LLM collaboration in creative tasks \cite{haghighi2025ontologies, shen2025societal, zhang2025noveltybench, doshi2024generative, vodrahalli2023artwhisperer, xie2024measuring, moon2024homogenizing, conitzer2024social}. Creative tasks, such as generating images or stories, are especially relevant given that there they are high-dimensional outputs that often no ``right'' answer. For example, users probably have much wider variability over their most-preferred art image, as compared to a factual query (e.g. ``What is the capital of New York State?''). 

Some empirical work focuses on when it may be appropriate to ask questions to clarify user intent: this has been studied both in traditional web search tools (e.g. \cite{zamani2020analyzing, keyvan2022approach}) as well as more recently in the context of LLMs \cite{wu2023new, andukuri2024star, zhang2024modeling, qiu2024can, li2023eliciting, kuhn2022clam, ma2024you, dutta2024problem, li2025questbench, kobalczyk2025active, wu2024aligning, kendapadi2024interact, dou2026gate, montazeralghaem2025asking}, as well as summary papers \cite{yuan2025query, deng2025proactive}.

Finally, our work is also closely connected with both the areas of feature acquisition and of fairness in recommender systems. Feature acquisition (e.g. see survey \cite{rahbar2025survey}) studies ways to acquire features in order to improve predictions: it differs from our work in that it does not consider generation and critique as an alternative framework, even for papers focusing on feature acquisition in generative AI (e.g. \cite{pmlr-v139-li21p}). For fairness in recommender systems, literature in this space studies tradeoffs related to fairness and utility in suggesting content to users (e.g. \cite{korevaar2023matched, li2023fairness, singh2018fairness}), or ensuring sufficient diversity among suggested items (e.g. \cite{kunaver2017diversity, peng2024reconciling, zhou2010solving, javari2015probabilistic}). Our paper differs from this body of work primarily in the mode of interaction: rather than suggesting content to users in a ranked list, in our setting users are allowed to interact directly with the algorithm (here, a LLM) and answer queries or critique generated content. Another difference is that recommender systems are typically selecting items from a pre-existing set of content to suggest to users, while our content creation setting implicitly assumes that the LLM can return every possible type of output. The main goal of this paper is to study how such modes of interaction affect the fairness and utility guarantees that we can give users. 

\section{Model}
\label{sec-model}
\subsection{Overview}

We consider the interaction between two different agents -- a user and an LLM.
Given an initial prompt from the user, the user has preferences over the LLM's
response space. In general, this preference distribution can be high-dimensional
and domain-specific. The preference space for emails is quite different from the
preference space for movie recommendations. For simplicity, we will take the
user's initial query as given.
 Moreover, we assume that the LLM knows the distribution of user preferences over the response space, while the user merely knows their own preferences. 

As an example of this communication asymmetry, consider the example of a user using an LLM to help plan their vacation given in Example \ref{ex:vacation}. The user has a sense of their own travel preferences (i.e. what locations they might want to visit), but not exact information about what's available. The LLM, on the other hand, has access to the full space of possible vacations but does not know the user’s specific desires. This mismatch creates a challenge for effective co-creation, as each party holds complementary but incomplete information.

In this model, the goal of the LLM is to generate outputs that are maximally aligned with a user's preferences. A key ingredient of our model is interactivity: We allow for the LLM to solicit additional information about the user by asking additional questions. While this is costly---users do not like to answer questions---it allows the LLM to make better inferences about the user's preferences.
The LLM must therefore trade off the cost of information acquisition with the benefits it brings to personalization.

\subsection{Preference Generation}

Here, we introduce a stylized model of preferences to ground our study.
Because the space of possible human preferences is vast, we focus on a simple, illustrative instantiation that captures the key properties of interest. For a given query, we represent each user's preferences as a vector of $n$ binary features, which we will refer to as their \emph{preference vector}.
A key property that our model captures is correlation between preferences: Information about the user's preferences along one dimension enables better inferences about their preferences along another. In Example \ref{ex:vacation}, for instance, users who prefer visiting location $A$ may also enjoy visiting location $B$, while users who like location $X$ may be more interested in also visiting location $Y$. In what follows, we will provide a minimal model that captures these relationships.

We assume that preference vectors are generated according to a particular Bernoulli mixture model e.g.,~\cite{bishop2006pattern}:

\begin{enumerate}
    \item 
    Users are randomly assigned to either cluster $0$ with probability $\alpha$ or cluster $1$ with probability $1 - \alpha$. Users assigned to cluster $0$ begin with all $0$ preference vectors, and users assigned to cluster $1$ begin with all $1$ preference vectors.
    \item 
    Each feature preference is flipped independently with probability $p < 0.5$.
\end{enumerate}
From this process, we can fully describe the probability that any preference vector occurs in the general population.

When thinking about the practicality of such a preference model, consider the
vacation planning example. The different sets of locations given in Example
\ref{ex:vacation} might correspond to high-cost vs. low-cost activities.
Travelers might generally plan their trips on a tight budget or seek luxury
experiences. However, those on a tighter budget might be open to splurging on one or two activities of interest, and those willing to spend more overall may still be looking to save in a few areas
Such behavior naturally fits into our model of preference generation, where cluster $0$ and cluster $1$ correspond to $n$ high-cost and low-cost activities, and any individual deviates their standard budgetary habits with some probability. Although incorporating correlated deviations across preferences could provide a richer model of user variation, assuming independent randomness adequately captures heterogeneity within preference clusters for the purposes of our analysis. 

More generally, this model is intended as a first step toward understanding the relationship between preference inference, information acquisition, and fairness; a full treatment of the richer preference structures that arise in specific application domains is left to future work.

\subsection{User-LLM Interaction}
Now that we have described the setup of feature preferences, we must define how the user and LLM exchange information. We consider an interaction model in which the user approaches the LLM with an initial task, but where the LLM does not know the human's full preferences over outcomes for that specific task (e.g. the full binary string). The $n$ features considered in the previous section are assumed not to be included in the initial task, and are unknown conditional on the task being defined.

To gain information about a user's preference, the LLM will directly query the user for their preference values of $k$ features, where $0 \leq k < n$. 

Based on the information learned through querying, the LLM then generates some output to give to the user. All $k$ of the feature values revealed by the user are assumed to match the LLM's output, leaving $n - k$ for the LLM to ``guess'' in its output. 

We assume that the only way an LLM can obtain information about the user's preferences is through querying individual features rather than directly querying a user's cluster identity. There are several possible reasons this assumption can make sense, including when clusters do not have well-defined names or due to user privacy concerns (i.e. a user may not want to reveal their exact financial status, as in Example \ref{ex:vacation}). 

The user's utility depends on (1) how many questions they had to answer, and (2) how close the generated output is to their true preferences.
We assume that the LLM has full knowledge of the preference distribution (i.e., parameters $p$ and $\alpha$).
Based on this assumption, we assume that the LLM's decision problem is comprised of both of the following:

\begin{definition}[Optimal response policy] Given a specific set of $k$ revealed features, the optimal response policy specifies how should the remaining features be filled in.
\end{definition}

\begin{definition}[Optimal querying policy] For a given $n$, the optimal querying policy determines the optimal number of features $k$ to query (assuming the optimal response policy fills in the remaining features).  
\end{definition}

\subsection{Utility}

Having guessed the user's preferences along $n-k$ dimensions, the LLM produces an output. We'll use $B$ to denote the disagreement between the output and the user's true preferences---that is, $B$ is the number of dimensions along which the LLM ``guesses wrong.''
Based on this quantity, we define a user's utility as

\begin{equation} 
    U(k, B) = n - k - B.
\end{equation}
In this notion of utility, the cost of a single query is taken to be the same as the benefit of identifying an additional preference correctly. Practically, the cost/benefit ratio does not have to be $1$, and we will explore other ratios empirically in Section \ref{sec-empirics}.

From this definition, the expected utility of a user is therefore a function
only on the expected number of mismatched features in the output compared to
their feature vector. Abusing notation slightly,
\begin{equation} \label{eq:util}
    \eu{k} = n - k - \mathbb{E}[B].
\end{equation}
The expected value of $B$ above is taken over both the possible sets of $k$
revealed bits and all possible user preference vectors. As mentioned in the
previous section, one of the LLM’s decision problem when maximizing utility is
to determine which $n - k$ bits to generate given the revealed bits and
preference generation parameters. As mentioned previously, the policy that
achieves the highest expected utility will be referred to as the optimal
response.

We include a complete formalization in \Cref{app-derivations}.

\subsection{Model Summary and Overview of Results}

To summarize, our model consists of two clusters of users with noisy and opposing preferences.
The LLM can choose to solicit information from the user, at some cost to user satisfaction.
Given the user's response, the LLM must infer the remainder of the user's preferences.

With this model, we begin by characterizing the LLM's welfare-maximizing policy in \Cref{sec-querying}.
How much information should it solicit from the user, and with that information, how should it make inferences about the user's unknown preferences?
The welfare-maximizing policy can lead to disparities in utility across users.
In \Cref{sec-fairness}, we show that the LLM's optimal response policy depends on the planner's aversion to inequality: By asking more queries, the LLM can personalize better and reduce inequality.
We consider the implications for content diversity in \Cref{sec-diversity}, showing that inequality-aversion can also lead to more diverse LLM outputs, reducing output homogeneity.
Before we conclude, we conduct an empirical validation of our model in
\Cref{sec-empirics}.

\section{When Asking a Question is Helpful}
\label{sec-querying}
We begin by asking two natural questions: what is the optimal response policy;
and how do we choose $k$, the optimal querying policy? In the first part of this
section, we will investigate how the optimal response policy is determined from
the set of revealed observations. Then, we move onto answering the second
question about determining the optimal number of queries to ask. We investigate
how initial cluster sizes, determined by parameter $\alpha$ in the preference
generation function, affect the optimal $k$ value. We observe that depending on
the probability $p$ with which each bit flips from its initial assignment,
the optimal strategy can vary: in some settings, it is best to ask no questions
at all, while in others, asking one or more queries yields higher utility.
Intuitively, more queries are helpful if they allow the LLM to make better
inferences about the remainder of the user's preferences.

\subsection{Optimal Response Policy}
First, we note that the optimal response policy is actually quite straightforward.

\begin{restatable}{lemma}{utilpolicy} 
\label{thm:util-policy}
    The optimal response policy first identifies the user’s cluster of origin, and then makes the most accurate prediction for the remaining features based on this cluster.
\end{restatable} 

Intuitively, because we consider a linear utility function, the
utility-maximizing choice is to guess the most likely element for each of the
user's preferences.
In the special case where $\alpha = 1/2$ and users are equally likely to start out with an all-$0$s or all-$1$s preference vector, the optimal response policy is simply to guess the bit value that is more frequently revealed during querying.

\subsection{Optimal Querying Policy}

Now, let us move on to considering the optimal $k$, the number of questions to ask.
This is determined by a trade-off between the value of acquiring information and its cost to user experience.

On the one hand, increasing $k$ allows the LLM to learn more about a user's
preferences, enabling better personalization.
On the other hand, soliciting information from the user requires effort on their part, which we model as costly.

\subsubsection{Equal-Sized Clusters}

To analyze $k$, we will begin by considering a simplified case of $\alpha$ (the
probability of being one cluster) $= 1/2$. In this setup, users are equally
likely to come from cluster 0 or 1.

\begin{restatable}{lemma}{queries}
\label{lem:queries}
    When $\alpha = 1/2$, for all $n > 2$, there exists a noise rate $p$ when it is optimal to ask a single query and a noise rate $p$ where it is optimal to ask no queries.
\end{restatable}

\begin{figure}[h!]
    \centering
    \includegraphics[width=.9\linewidth]{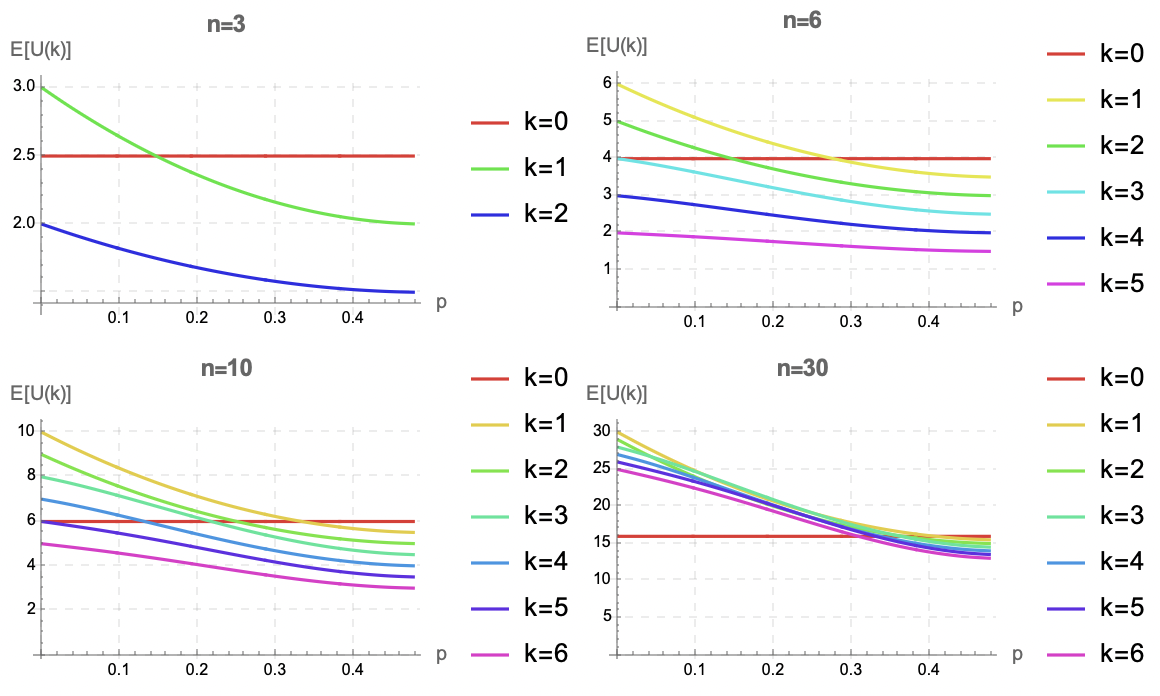}
    \caption{Expected utility $\eu{k}$ vs $p$ in the case of $\alpha = 1/2$ for $n = 3, 6, 10$, and $30$. For each $n$, as the noise rate $p$ increases, there comes a point where it optimizes expected utility to not ask any questions. In other words, as there is more uncertainty to the user's cluster of origin based on the information, there is less value to asking questions. Additional figures are provided in Appendix \ref{app-plots}.}
    \label{fig:equal-clusters}
\end{figure}

When the noise rate $p$ is low, a single query often suffices to identify the user’s preference cluster with high probability. Once the cluster is known, the remaining bits are likely aligned with it, making subsequent predictions more accurate and the expected utility higher.

Conversely, as $p$ approaches $1/2$, each bit in a user’s preference vector becomes nearly random. Recovering the user's original cluster is thus both \emph{harder} and \emph{less informative}. Any query yields mostly noise, so the LLM’s guesses are no better than random while still incurring a utility cost for $k > 0$. Hence, when $p$ is near $1/2$, it is often better to set $k = 0$.

Together, these cases illustrate that while soliciting information can substantially improve inference when the signal is reliable, it becomes counterproductive when the signal is dominated by noise. In general, querying is valuable only when the information obtained is sufficiently strong to offset its cost.

\subsubsection{Unequal Clusters}

Next, we consider the more general case of clusters occurring with probability $\alpha$ and $1 - \alpha$. In light of Example \ref{ex:vacation}, assume that there's a minority group that is more cost-sensitive. Without additional information, the LLM's ``best guess'' would be to satisfy the majority group's utility, which would lead to a more expensive itinerary. Soliciting more information in terms of preferences could make it more likely to satisfy the minority group's preferences. However, depending on the noise rate, the additional cost in terms of querying may not be ``worth it'' when the dominant group is large. 

\begin{restatable}{lemma}{extremealpha}
\label{lem:extreme-alpha}
    There exist values of $n$ where it maximizes expected utility to ask no questions for all values of $p$.
\end{restatable}

\begin{figure}[h!]
    \centering
    \includegraphics[width=.9\linewidth]{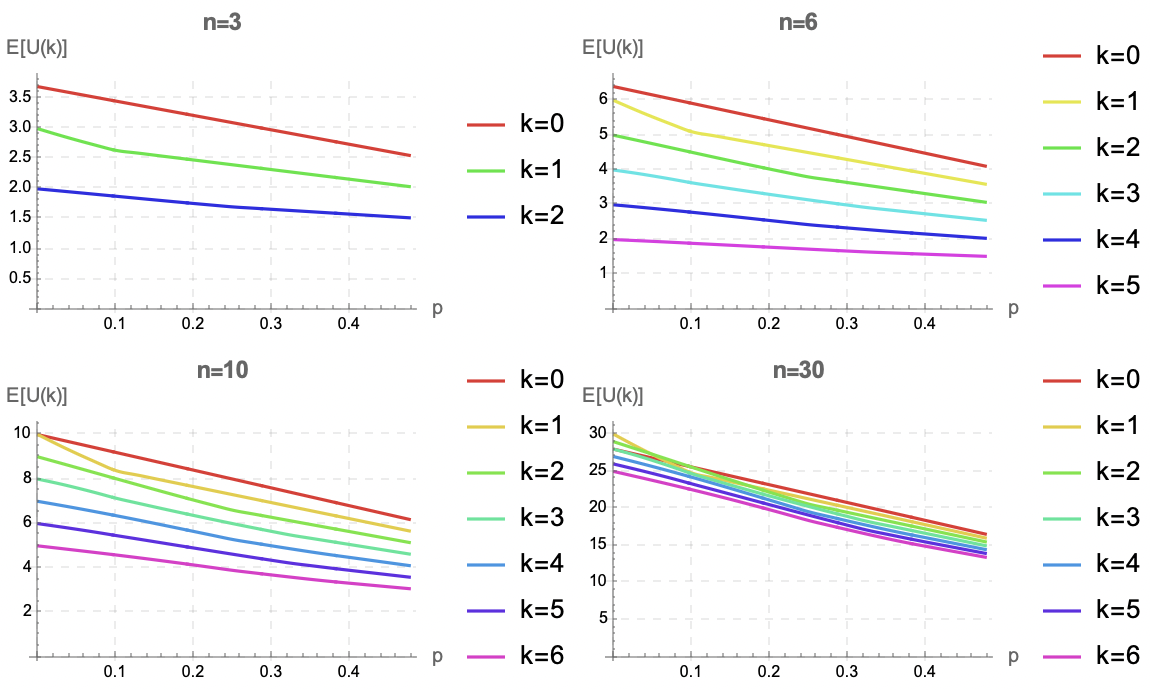}
    \caption{Expected utility $\eu{k}$ vs $p$ in the case of $\alpha = 0.9$ for $n = 3, 6, 10$, and $30$. When $n$ is smaller, $k = 0$ is the optimal querying policy regardless of $p$. When $n$ is larger however, it can sometimes be better to ask a single question.}
    \label{fig:extreme-alpha}
\end{figure}

When $\alpha$ is either close to $0$ or $1$, the niche set of users is significantly smaller than those with majority preferences, leading the model to have a much stronger prior that any user is part of the majority. When this prior belief is sufficiently strong, the utility cost paid from asking a single query might outweigh any information gained from the query. This can even be true for small noise rate $p$ if the group of users with niche tastes is sufficiently small.

Figure \ref{fig:regime} provides an overview of the landscape of optimal $k$ values for a range of choices for $\alpha \in [0.5, 1]$ and $p \in [0, 0.5]$ for $n = 30$. We see here that for increased values of $n$, it can be useful to ask more than a single question, particularly for intermediate values of $p$ and $\alpha$. In these cases, the prior belief about cluster membership is strong but not decisive, so the additional clarification provided by a second query meaningfully increases confidence in the remaining preference inferences.

\begin{figure}[h!]
    \centering
    \includegraphics[width=.5\linewidth]{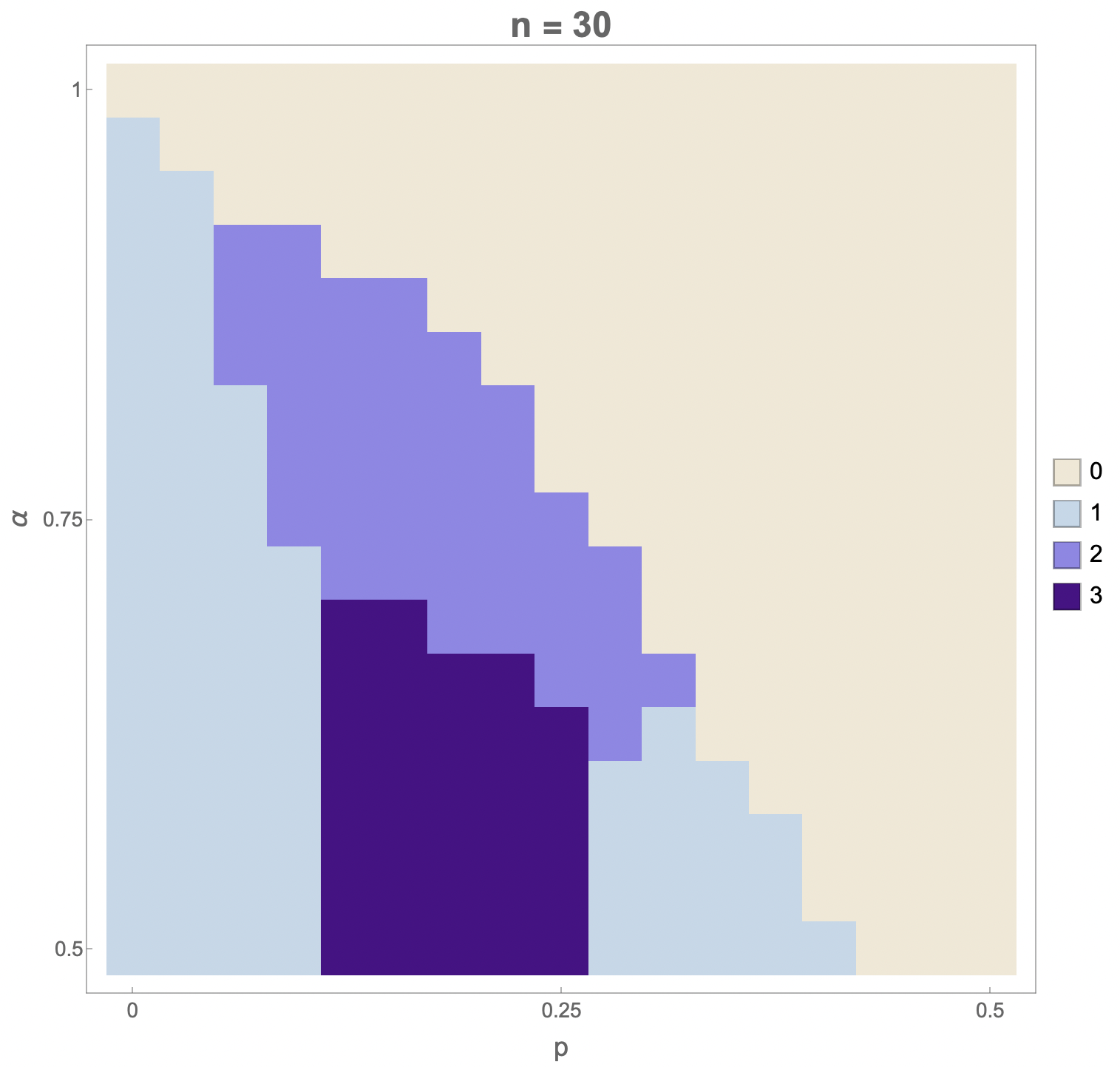}
    \caption{Optimal number of queries $k$ for various values of noise rate $p$ and probability of cluster $0$ membership $\alpha$ with $n = 30$. The plot is symmetric for smaller values of $\alpha$ from $0$ to $0.5$.}
    \label{fig:regime}
\end{figure} 

However, these results also make clear that under expected utility maximization, information acquisition is driven entirely by aggregate gains. When one preference cluster is sufficiently small, the model optimally forgoes clarification because the expected benefit is too low. As a result, minority users are systematically under-served: their preferences are unlikely to be inferred correctly, even when the noise level is small. This raises concerns of fairness, as minority users have little to no chance of receiving outputs that align with their preferences. In the following section, we explore alternative objective functions that explicitly account for this aversion to inequality.

\section{Fairness and Welfare Considerations}
\label{sec-fairness}
The prior section studied how to optimize for expected utility over a population. In this section, we turn to questions of fairness across a population. We study how both the optimal response policy and optimal querying policy change if we are optimizing for different notions of welfare.

\subsection{Welfare}

We model welfare using a single-parameter family of social welfare functions for any output and where the parameter $\gamma$ represents the level of aversion to inequality. Within the context of this paper, we will be focusing on welfare functions that marginalize over all possible outputs given possible sets of $k$ revealed features.
\begin{equation}
\label{eq:welfare}
    W(k) = \frac{1}{1 - \gamma} \mathbb{E}\left[\sum_{\text{users}} \text{Utility(user)}^{1 - \gamma}\right].
\end{equation}
Note that our social welfare function captures expected ex-post
welfare---i.e.,~measured over \textit{realized} outcomes, as opposed to ex ante
expected outcomes. Intuitively, this is the difference
between $E[U^{1-\gamma}]$ vs. $E[U]^{1-\gamma}$, where $U$ represents a user's utility.

The exact formulation of the social welfare function we consider in this setup can be found in Appendix \ref{app-derivations}. Overall, this family of social welfare functions are widely used in related literature in computer science (e.g. \cite{pardeshi2024learning, cousins2023revisiting, qi2000new, hendrycks2025introduction}) and economics (e.g. \cite{atkinson1970measurement, blackorby1978measures}). Choosing such an objective function is desirable because it directly encodes trade-offs between equity and efficiency, which also provides a parameter that can enable alignment with different ethical or policy priorities.

There are several natural choices of $\gamma$ with such a welfare function: $\gamma = 0$ yields a utilitarian welfare function (as was covered in the previous section), $\gamma = 1$ yields the Nash welfare (the product of utilities), and $\gamma = \infty$ implies maximin fairness.\footnote{More specifically, the Nash welfare is given by the limit of~\eqref{eq:welfare} as $\gamma$ approaches 1.}

\subsection{Optimal Response Policy}

In principle, the LLM's optimal response might be randomized. Given information
provided by the user, the LLM's welfare-maximizing response could be randomized.
In fact, we show that this is not the case in our setting. Deterministic
responses suffice:

\begin{restatable}{lemma}{det}
\label{lem:det}
    In order to maximize ex-post social welfare, the best course of action for an LLM is to deterministically return a single response given the revealed features.
\end{restatable} 

The intuition is straightforward: assume by contrast that the LLM used a randomized strategy. Users experience welfare based on the specific guesses made given revealed features. Because we are using ex-post social welfare, our utility over a randomized set of responses is simply equal to the expected utility, with randomness taken over our responses. Out of all possible guesses given revealed features, there will be a single one that maximizes welfare. By linearity of expectation, using a deterministic strategy of only returning those guesses given $k$ revealed features has welfare at least as high as the original randomized strategy, and strictly higher in every case without ties in welfare between items. 

Next, we consider the optimal response policy. Beyond the $\gamma = 0$ case
discussed in Section \ref{sec-querying}, interesting things start to happen to
the optimal response policy as we increase $\gamma$. The optimal response policy
for the purely utilitarian setting is no longer relevant as we prioritize
fairness, and is now a function of $\gamma$ as well. Let $a$ be the number of
1's in response to the LLM's $k$ initial queries. Then, denote the optimal
number of 1's in the LLM's final output as $f_\gamma(a)$.

\begin{restatable}{lemma}{optimalpolicy}
\label{lem:opt-policy}
    The optimal response policy may produce a mixed output of $0$s and $1$s.
\end{restatable} 

In contrast the results given in Lemma \ref{thm:util-policy}, the optimal
inequality-aware policy is to hedge, rather than simply cater to the majority
tastes. As a result, the LLM will make a combination of choices that favor
each cluster.

\begin{figure}[h!]
    \centering
    \includegraphics[width=.9\linewidth]{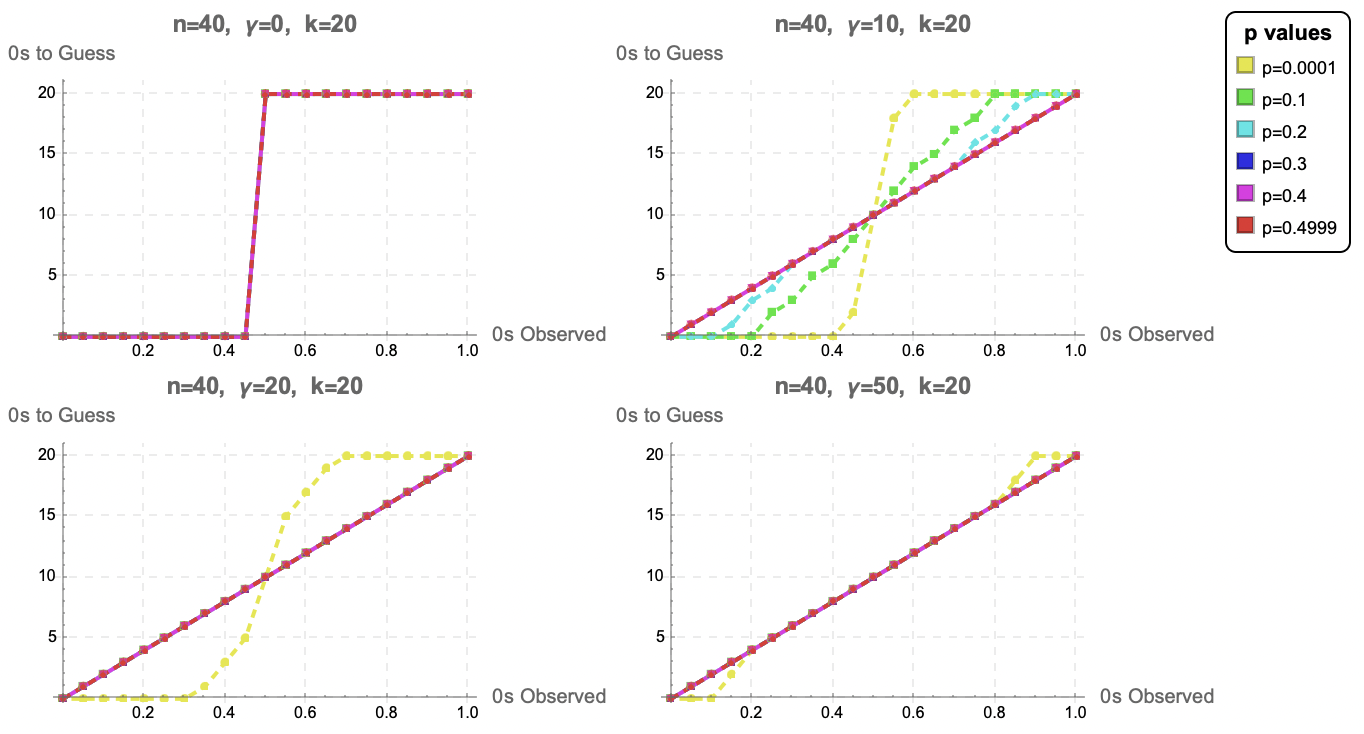}
    \caption{The number of $0$s to output based on the optimal response policy $n - k - f_\gamma(a)$ as a function of $a$ for various $p$ values when $\alpha = 1/2$, $n = 40$, and $k = 20$. As $\gamma$ increases, the optimal policy for all noise rates gets closer to being linear in the composition of the observations.}
    \label{fig:optimal-policy-short} 
\end{figure}

As $\gamma$ becomes increasingly large, however, the policy does not simply become outputting half $0$s and half $1$s to hedge its guesses towards both clusters. Rather, it reflects a trade-off between guessing equal numbers of $0$s and $1$s and making use of the information learned from the revealed bits. In contrast to the result of Lemma \ref{thm:util-policy} where much of the solicited information may not directly impact the optimal response policy, the LLM actually makes use of all information learned from the revealed bits.

\begin{restatable}{theorem}{maximinpolicy} 
\label{thm:maximin-policy}
    For every fixed $k \leq n/2$, each set of revealed preferences $a$ results in a distinct optimal response policy $f_\gamma(a)$ for a maximin fairness objective ($\gamma \rightarrow \infty$). 
\end{restatable} 

For some intuition about why this is true, consider that maximin utility is maximized in the setting when there are the fewest number of users who are maximally worst-off (i.e. receive utility $0$). Moreover, users with the lowest utility are those that have exact opposite values for every term guessed by the LLM. As a result, this tells us that an LLM should guess the remaining bits to be the opposite of the least-likely user, given the bits we've observed. We thus arrive at the somewhat surprising result that an LLM should sometimes produce responses that vary sharply with the responses given when the goal is maximin fairness, depending on the preferences of the least-likely user, as we can notice in the rightmost panel of Figure \ref{fig:optimal-policy-short}. Importantly, this also clarifies why greater inequality aversion increases the value of information: the more sensitive the objective is to worst-case outcomes, the more the model directly benefits from information in queries. While this explains how the model should use information once it is obtained, it does not yet address the separate question of how much information should optimally be elicited in the first place, which we will consider next.

\subsection{Optimal Querying Policy}

Another interesting question to consider when generalizing our setup to prioritize fairness is how this affects the optimal number of queries $k$ that the LLM should ask the user before generating content. 

Intuitively, one might think that it is best to ask $n$ queries and uniquely identify each user's preferences, making sure everyone is equally happy. However, while this minimizes inequality, it does not maximize the inequality-averse social welfare we consider, as all users receive utility $0$ in this plan. Under an inequality-averse social welfare objective, the goal is to balance minimizing inequality with maximizing total utility, leading to the adoption of different strategies.

\begin{restatable}{lemma}{kdecreasing}
\label{lem:k-decreasing}
    The optimal number of queries is non-monotonic as a function of $\gamma$ for $p > 0$.
\end{restatable} 

\begin{figure}[h!]
    \centering
    \includegraphics[width=.9\linewidth]{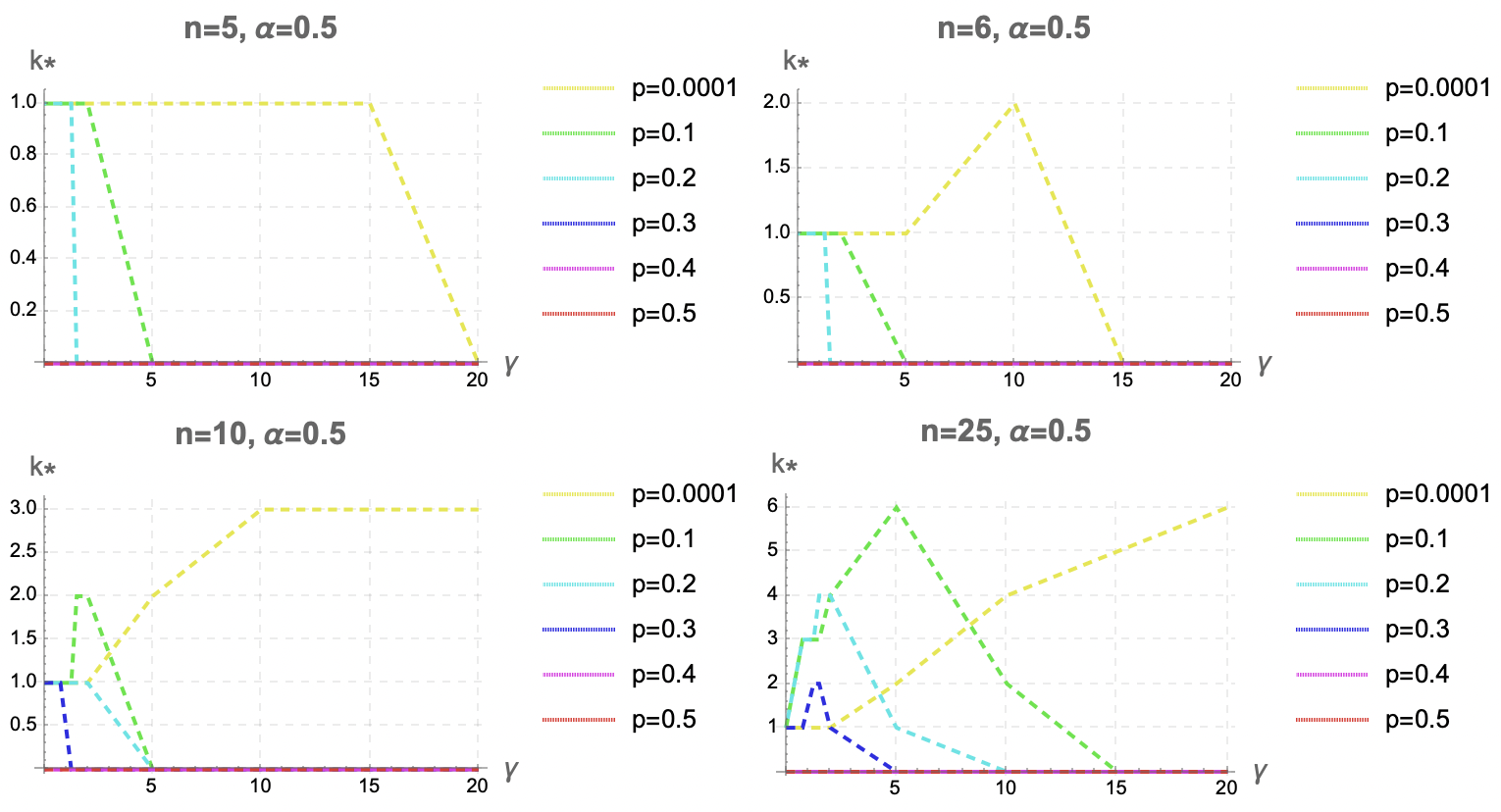}
    \caption{Optimal choice of $k$ vs $\gamma$ for various values of $p$ and $\alpha = 1/2$. As $\gamma$ initially increases, the optimal querying policy grows for small noise rates $p$. As $\gamma$ continues to grow however, the optimal querying policy begins to decrease again.}
    \label{fig:optimal-queries-short}
\end{figure}

This observation is not immediately clear. As mentioned above, one might expect
that asking more questions would always be beneficial, as it provides the LLM
with additional information and therefore a higher likelihood of correctly
inferring the missing bits. However, this potential information gain must be
weighed against the cost associated with each query. When the noise rate $p$ is
very small, increasing the number of queries initially improves performance as
$\gamma$ increases, as the LLM can more reliably identify the user’s cluster of
origin. In Figure \ref{fig:optimal-queries-short}, this can be seen for the
choices of flipping probabilities $p = 0.1$ and $p = 0.0001$, as the choice of $k$ initially increases for $n = 10$ and $25$. However, for any nonzero $p$,
there comes a value of $\gamma$ where the marginal benefit of additional queries
is outweighed by their cost. This trend can be seen in Figure
\ref{fig:optimal-queries-short} when for noise rate $p = 0.1$, we see the
optimal $k$ fall back down to $0$ as $\gamma$ continues to grow for $n = 10$ and
$n = 25$. As $\gamma \rightarrow \infty$, the optimal strategy switches to be
simply guessing all bits of a user's preference vector. Intuitively, for a
maximin welfare objective, information does not help.

\section{Diversity and Homogeneity in Outputs}
\label{sec-diversity}
Finally, we turn to the connection between fairness and output diversity.
We might intuitively expect a positive relationship between these: A preference for fairness should lead to a higher emphasis placed on personalization, thus reducing output homogenization.
While this can be true in specific instances, our main result in this section will be to show that the intuition given above does not hold in general.

Recall that by \Cref{lem:det}, the LLM's optimal response policy is a deterministic function of the strength of its belief that the user is from cluster 0 vs. cluster 1.
As defined above, $f_\gamma(a)$ specifies how many 1's are in the LLM's inferred
preference vector, given the user's response $a$ to the LLM's initial queries.
To quantify content heterogeneity, we will characterize how the entropy over $f_\gamma(a)$ varies as $\gamma$ increases.
Formally, this is $\mathbb{E}[H(f_\gamma(a))]$, where $H$ is the standard Shannon entropy and the randomness is over the user's response $a$.
Note that the distribution of $a$ depends on the number of queries $k$ asked by the LLM.
Here, we assume the LLM follows the optimal query policy for that value of $\gamma$. With this, we can state our final result:
\begin{restatable}{lemma}{entropydecrease}
\label{lem:entropy-decrease}
    $\mathbb{E}[H(f_\gamma(a))]$ (the entropy of responses generated by the optimal query and response policies, given a particular $\gamma$) is non-monotonic in $\gamma$. 
\end{restatable}

To provide some intuition, consider what happens when $\gamma$ grows from 0 to $\infty$.
For $\gamma = 0$, by \Cref{thm:util-policy}, the LLM infers the user's cluster of origin and caters exclusively to that cluster.
This effectively yields two possible outputs, based on the inferred cluster of origin, and low output entropy.
As $\gamma$ increases, following the optimal query policy, the LLM may ask more queries and receive more information from the user.
Because the LLM has more information to work with, it can produce a broader set of possible outputs.
For example, when asking $k=2$ queries, there are 3 categories of responses that the LLM can get: all 0's, all 1's, or one 0 and one 1.
By providing different responses in each of these cases, the LLM's output distribution spans a larger set of potential outputs.
And as we show in our proof, there are instances where this initially leads to higher output diversity.
However, as shown in \Cref{fig:optimal-queries-short}, the optimal number of queries eventually drops for sufficiently large $\gamma$.
As a result, the LLM no longer solicits much information from the user, meaning that the diversity of its output distribution must go down.

In a sense, this result is quite related to the data processing inequality: the entropy of the output distribution of a deterministic function is upper-bounded by the entropy of its input distribution.
With little information about an individual user's preferences, the LLM cannot offer much by way of output diversity (given that the optimal response policy is a deterministic function of the inputs.)
Given more queries, as long as the optimal response policy is sufficiently sensitive to user-supplied information, output diversity increases.

These findings show that there is no simple relationship between optimizing for fairness and optimizing for output diversity.
The relationship between fairness and output homogenization is mediated by the information-soliciting behavior of the LLM.
If fairness leads to an increase in more information solicited from users, we should expect to see a corresponding increase in output diversity.

\section{Empirics}
\label{sec-empirics}
The theoretical framework developed above relies on a highly stylized preference
model. A natural question is whether the insights carry over to real preference
data. Our primary goal with this section is to demonstrate that the optimal
number of queries to ask depends on one's level of inequity aversion. 

\subsection{Experiment Setup}

We consider the sushi preference dataset introduced in \cite{kamishima2003nantonac, kamishima2006efficient}, which contains the complete rankings of 10 different sushi by 5000 users. To convert ranking data to binary preference data, we encode the preference for any sushi ranked in the top 2 as a $1$, and any sushi ranked in the bottom 8 as a $0$. 

In order to map to our model, the goal of our experiment is to identify a user's preferences on each of the 10 pieces of sushi rather than to just recommend sushi they'd like. To simulate preference elicitation, we began by training a decision tree classifier where each split is chosen to maximize information gain about user preferences over the 10 sushi types. For a given $k$, a user's preferences for the first $k$ sushi types following tree order are revealed. The remaining $10 - k$ preferences are selected to be the most likely values conditional on a user's revealed preferences. 

In a general setting, it makes sense to model differing costs of queries. Hence, to allow for more flexibility, we modify our utility function to include an explicit per-query cost $c$, leading to a revised function of 
\begin{equation}\label{eq:cost-utility}
    U(k, B) = n - c \cdot k - B,
\end{equation}
where $B$ is the number of preferences guessed incorrectly, $k$ is the number of
queries, and $n = 10$. We will use this generalized utility function in our
definition of welfare and then analyze the optimal $k$ for different values of
$\gamma$.

\subsection{Fairness-Utility Tradeoff}

Figure \ref{fig:sushi-welfare} demonstrates how welfare changes as the optimal number of queries $k$ increases for different values of the inequality aversion parameter $\gamma$. For purely utilitarian welfare, we can see that it is not until $k = 3$ queries that the preference information learned is sufficiently valuable to increase welfare. Moreover, welfare becomes optimal at $k = 4$ and quickly begins to trend down from its base level as the number of queries increases further. 

When we become more inequality-averse and increase $\gamma$, we see that the
optimal $k$ changes. When $\gamma$ is small, we see that the optimal $k$ remains
at $4$. However, when $\gamma = 5$, we
see that the optimal $k$ increases to $9$. As $\gamma$ further increases to
$10$, we see that it becomes optimal for welfare to directly elicit a user's
entire set of preferences. As in our theoretical setup, the optimal information
elicitation behavior depends on our level of inequality-aversion $\gamma$.

\begin{figure}
    \centering
    \includegraphics[width=\linewidth]{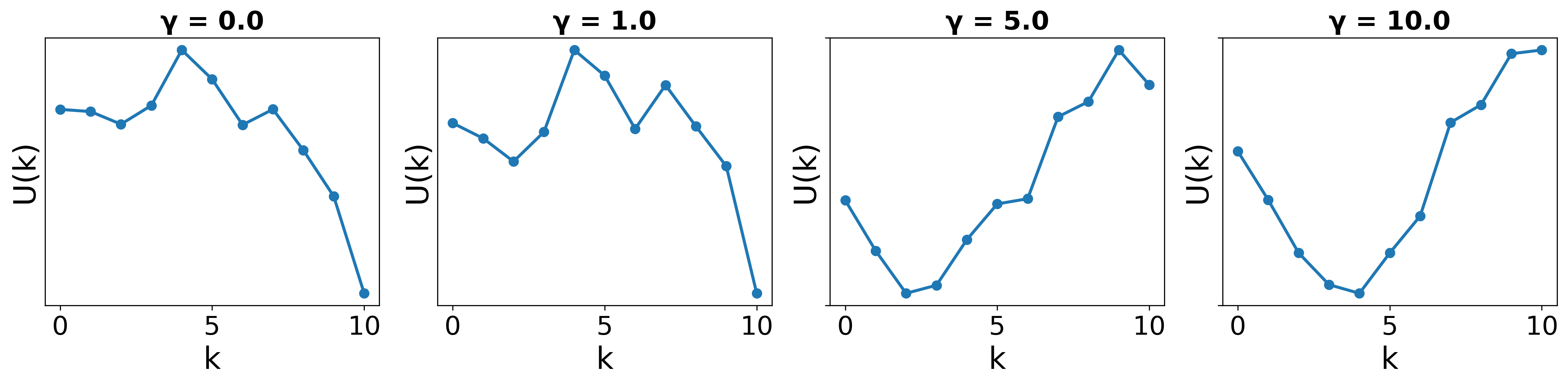}
    \caption{Mean welfare vs $k$ for increasing values of $\gamma$ for per-query cost $c = 0.3$. As $\gamma$ increases, the optimal number of queries $k$ changes, increasing from $4$ to $10$ as $\gamma$ increases from $0$ to $10$. The $y$-axis scale is omitted, as welfare is not comparable across $\gamma$; we consider only within-$\gamma$ comparisons over $k$.}
    \label{fig:sushi-welfare}
\end{figure}

Figure \ref{fig:sushi-tradeoff} illustrates the underlying fairness–utility tradeoff. Setting $k = 4$ maximizes mean utility but generates substantial inequality by favoring majority preferences. At the other extreme, fully eliciting preferences minimizes inequality but reduces utility due to query costs. Notably, because the marginal benefit of querying exceeds the cost in this setting, utility at $k = 10$ remains strictly positive. Consequently, as $\gamma \to \infty$, higher levels of querying are always preferred.

\begin{figure}
    \centering
    \includegraphics[width=0.5\linewidth]{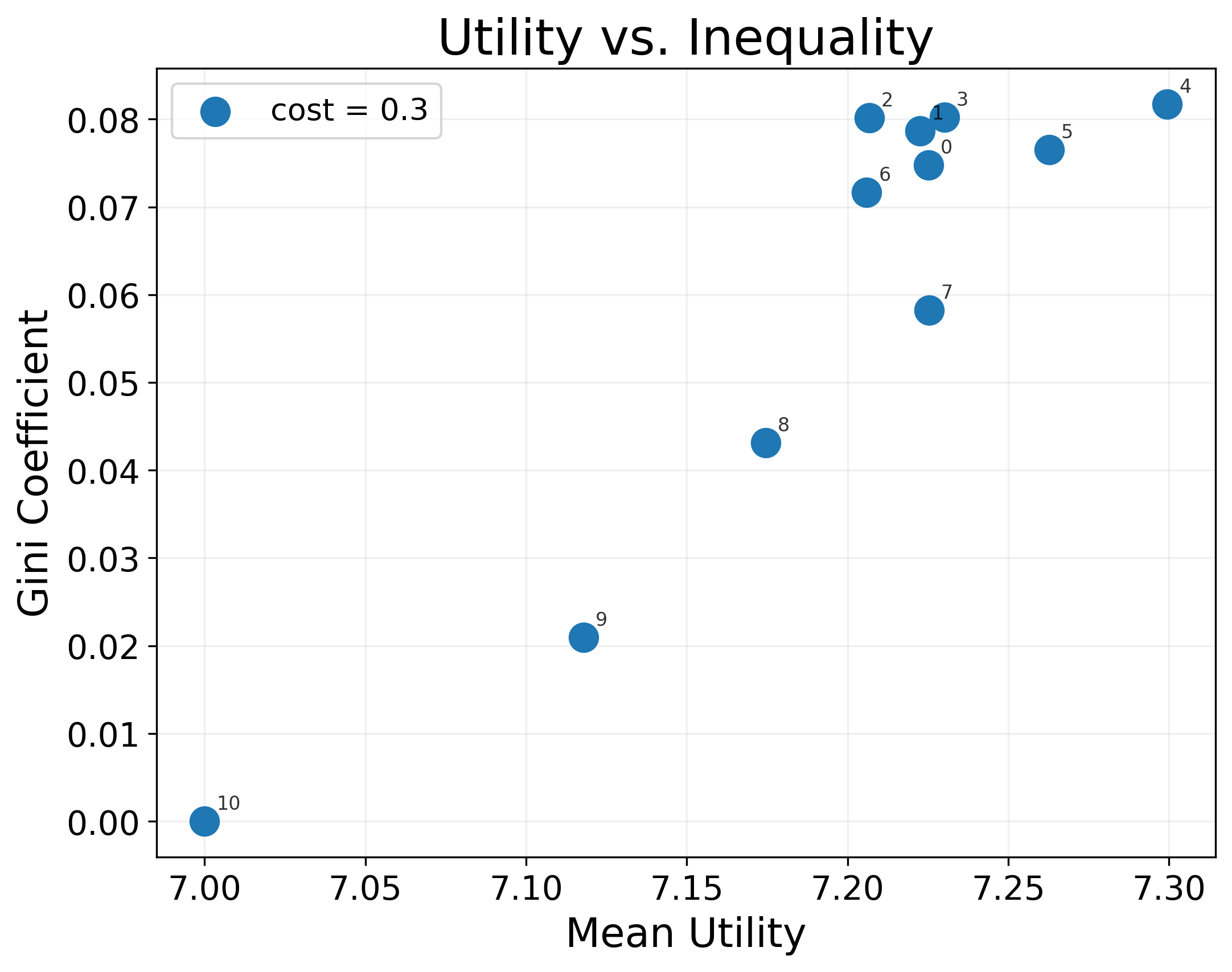}
    \caption{Mean utility vs Gini coefficient for different values of $k$ and $c = 0.3$. Smaller $k$ are clustered in the upper-right, achieving higher mean utility at the cost of greater inequality. As k increases toward $10$, mean utility declines but so does inequality, with k = 10 attaining near-perfect equality.}
    \label{fig:sushi-tradeoff}
\end{figure}

\section{Discussion}
\label{sec-discussion}
Our work introduces a mathematical framework for interaction between humans and LLMs, formalized under a particular model of preference generation. Within this framework, we assume that the LLM is able to obtain information about user preferences by directly querying users for their preferences -- a behavior that is starting to be utilized by LLM chatbots in practice. We also consider an objective function that maximizes social welfare, with a parameter $\gamma$ that controls the trade-off between fairness and utility. Despite its simplicity, our model reveals subtle relationships between fairness, utility, and output diversity---relationships that may practically emerge even when model assumptions are violated.

One key limitation of our model is its reliance on a highly stylized
distribution of user preferences. The Bernoulli mixture model fits some data
well but may not generalize to other domains. While our experiments in
\Cref{sec-empirics} validate some of our findings, further
empirical testing is needed to assess
when querying strategies are valuable for a specific domain. Moreover, the model
assumes that an LLM has perfect knowledge of the full distribution of user
preferences — however, this assumption may be imperfect, given that LLMs may rely on incomplete or biased data (our experiments in Section \ref{sec-empirics} explore relaxations of this assumption). Finally, the model treats all users as facing the same cost when
answering queries, though in reality this burden varies with each individual’s
circumstances and the specificity of the question.

Future work could generalize our model in a number of ways.
One could consider more general preference models, differential costs for information, multi-round conversations, and stateful personalization.
More generally, information solicitation by LLMs is fruitful direction for future research, and we seek to contribute to the literature with theoretical grounding that might help guide and motivate future empirical work.

\printbibliography

\appendix

\section{Derivations}
\label{app-derivations}
\subsection{Utility}

To obtain a formalization of expected utility from the high-level form given in Equation \ref{eq:util}, we start by breaking down the expected value of utility into 2 cases -- those where the original user is from cluster $0$ and those where the original user is from cluster $1$. 


Recall that $f(a)$ is the exact number of $0$s guessed by the LLM's policy after observing $a$ revealed $1$s. All remaining guesses are $1$s. In the case where the original user is from cluster $0$, any mistakes within the guessed $0$s would be bits that flipped from the original cluster $0$ to a $1$, a proof of which can be seen in Appendix \ref{app-proofs}. Similarly, any mistakes within the guessed $1$s would be bits that did not flip from the original cluster $0$. Let $j$ be a counter for the number of true $0$s that were guessed to be $1$s and $\ell$ be a counter for the number of true $1$s that were guessed to be $0$s. Using our noise rate, we can write this as

\begin{equation*}
    \Pr(\text{$j$ mistakes in $0$s $\mid$ cluster $0$}) = \binom{f(a)}{j}(1-p)^{f(a) - j} \cdot p^j.
\end{equation*}
Similarly, we can express
\begin{equation*}
    \Pr(\text{$\ell$ mistakes in $1$s $\mid$ cluster $0$}) = \binom{n - k - f(a)}{\ell} (1-p)^\ell \cdot p^{n - k - f(a) - \ell}.
\end{equation*}
Summing over all possible choices of $j$ and $\ell$, we can generalize this to write the probability of any mistake given $a$ $1$s were revealed as
\begin{equation*}
    \Pr(\text{mistake $\mid$ cluster $0$}) = \sum_{\substack{j \leq f_\gamma(a) \\ \ell \leq n - k - f_\gamma(a)}} \binom{f_\gamma(a)}{j}(1-p)^{f_\gamma(a) - j}p^j \cdot \binom{n - k - f_\gamma(a)}{\ell} (1-p)^\ell p^{n - k - f_\gamma(a) - \ell}.
\end{equation*}
From the previous equation, note that the total number of mistakes $j + \ell$ is actually the sum of two Bernoulli RVs with different parameters. Thus, we use linearity of expectation to simplify the expected number of mistakes $\mathbb{E}[B]$ from Equation \ref{eq:util} conditioned on a user belonging to cluster $0$ and taken with respect to all possible user preference vectors as
\begin{equation*}
    \mathbb{E}[B \mid \text{cluster 0, $a$ revealed $1$s}] = f(a) \cdot p + (n - k - f(a))(1-p).
\end{equation*}
We now must take expectation with respect to all possible counts $a$ of revealed $1$s within the $k$ revealed features. Marginalizing over the value of $a$, we find
\begin{equation*}
    \mathbb{E}[B \mid \text{cluster $0$}] = \sum_{a = 0}^k \binom{k}{a} (1-p)^{k-a}p^a \cdot (f(a) \cdot p + (n - k - f(a))(1-p)).
\end{equation*}
Exploiting the symmetry of the preference generation process for the case where a user originates from cluster $1$, we can recognize that the expressions are identical except for the values of $p$ and $(1-p)$ being swapped. Hence, we can write
\begin{equation*}
    \mathbb{E}[B \mid \text{cluster $1$}] = \sum_{a = 0}^k \binom{k}{a} (1-p)^ap^{k-a} \cdot (f(a) \cdot (1-p) + (n - k - f(a)) \cdot p).
\end{equation*}
Marginalizing once more over user cluster of origin, we are able to write
\begin{equation*}
    \mathbb{E} = \alpha \cdot \mathbb{E}[B \mid \text{cluster $0$}] + (1 - \alpha) \cdot \mathbb{E}[B \mid \text{cluster $1$}].
\end{equation*}
Finally, plugging everything back in, we end up with the formalized expression of expected utility as
\begin{align*} \label{eq:full-util}
    \eu{k} = n - k &- \alpha \cdot \sum_{a = 0}^k \binom{k}{a} (1-p)^{k-a}p^a \cdot (f(a) \cdot p + (n - k - f(a))(1-p)) \\
    &- (1 - \alpha) \cdot \sum_{a = 0}^k \binom{k}{a} (1-p)^ap^{k-a} \cdot (f(a) \cdot (1-p) + (n - k - f(a)) \cdot p).
\end{align*}

\subsection{Welfare}

To obtain an expression for welfare based on the general form given in Equation \ref{eq:welfare}, we can use much of the same logic as in the previous section. The main practical differences between the expression for expected utility and the expression for welfare is the inclusion of the $\gamma$ parameter which the utility is raised to. Due to the presence of this parameter, we can no longer apply linearity of expectation to pull out the $n - k$ in the expression for $W(k)$. Hence, we must directly plug in the full expression for utility into our summation.

As in the previous derivation, begin by considering the case where a user comes from cluster $0$. As before, if we assume $a$ $1$s have been observed, we have
\begin{equation*}
    \Pr(\text{mistake $\mid$ cluster $0$}) = \sum_{\substack{j \leq f_\gamma(a) \\ \ell \leq n - k - f_\gamma(a)}} \binom{f_\gamma(a)}{j}(1-p)^{f_\gamma(a) - j}p^j \cdot \binom{n - k - f_\gamma(a)}{\ell} (1-p)^\ell p^{n - k - f_\gamma(a) - \ell}.
\end{equation*}
Recall that the total number of mistakes in this equation is equal to $j + \ell$, so the utility for any given $j$ and $\ell$ is $n - k - j - \ell$. For simplicity of notation, define
\begin{align*}
        G_{0, \gamma}(z, p)
        &\triangleq
        \sum_{\substack{j \leq z \\ \ell \leq n - k - z}} \binom{z}{j}(1-p)^{z - j}p^j \cdot \binom{n - k - z}{\ell} (1-p)^\ell p^{n - k - z - \ell} \cdot (n - k - j - \ell)^{1 - \gamma}.
\end{align*}
From here, let $W_0(k, \gamma)$ denote the welfare given that a user is from cluster $0$. With these definitions and conditioned on seeing $a$ $1$s, we arrive at the following welfare taking the form of our family of welfare functions
\begin{equation*}
    W_0(k, \gamma \mid \text{$a$ revealed $1$s}) = \frac{1}{1 - \gamma} \cdot G_{0, \gamma}(f_\gamma(a), p).
\end{equation*}
Marginalizing over all possible sets of $k$ revealed features, we are able to write conditional welfare taken over the randomness of $k$ as
\begin{equation*}
    W_0(k, \gamma) = \frac{1}{1 - \gamma} \cdot \sum_{a = 0}^k \binom{k}{a}(1-p)^{k-a}p^a \cdot G_{0, \gamma}(f_\gamma(a), p).
\end{equation*}
Similarly to before, we can exploit the symmetry of the preference generation process for the case where a user originates from cluster $1$, and express the conditional welfare for cluster $1$ as
\begin{equation*}
    W_1(k, \gamma) = \frac{1}{1 - \gamma} \cdot\sum_{a = 0}^k \binom{k}{a}(1-p)^a p^{k-a} \cdot G_{1, \gamma}(f_\gamma(a), 1-p).
\end{equation*}
This expression is nearly identical to $W_0(k, \gamma)$, except for the $p$ and $1 - p$ values are flipped, in line with the different starting clusters. 

Finally, marginalizing once more over cluster of origin, we are able to get a full expression for welfare as
\begin{equation}\label{eq:full-welfare}
    W(k, \gamma) = \frac{1}{1 - \gamma} \parens*{\alpha \cdot W_0(k, \gamma) + (1 - \alpha) \cdot W_1(k, \gamma)}.
\end{equation}

\section{Proofs}
\label{app-proofs}
\subsection{Proof of Lemma \ref{thm:util-policy}}
\utilpolicy*
\begin{proof}
    In order to prove the optimal response policy in the $\gamma = 0$ case, consider the basic definition of expected utility given in Equation \ref{eq:welfare} as
    \begin{equation*}
        \eu{k} = n - k - \mathbb{E}[b],
    \end{equation*}
    where $b$ is the number of bits that were guessed incorrectly. The expected number of mistakes can be written as the sum of the probabilities of guessing each bit incorrectly. Moreover, as all bits are flipped during our preference generation process independently with the same probability $p$, we can express the expected number of mistakes as simply
    \begin{equation*}
        \mathbb{E}[b] = (n - k) \cdot \Pr(\text{single bit is guessed incorrectly}).
    \end{equation*}
    Thus, maximizing expected utility is equivalent to minimizing the probability of an incorrect guess for each bit. 
    
    Because all bits are identically distributed and independent, the same optimal decision rule applies to each bit. The optimal guess for a bit is therefore the value (either $0$ or $1$) that maximizes its posterior probability given the revealed information. Using Bayes’ rule, this corresponds to choosing the more likely bit value conditional on the observed revealed bits. Hence, we can claim that the optimal response policy in the $\gamma = 0$ case is to always guess the same value based on revealed bits.
\end{proof}

\subsection{Proof of Lemma \ref{lem:queries}}
\queries*
\begin{proof}
    Consider the plot given in Figure \ref{fig:equal-clusters}. We can see that for $n = 3$, it is best to ask a single query when $p = 0$. However, in that same $n = 3$ plot when $p = 0.5$, it is best to ask no queries in order to maximize utility.
\end{proof}

\subsection{Proof of Lemma \ref{lem:extreme-alpha}}
\extremealpha*
\begin{proof}
    Consider the plot given in Figure \ref{fig:extreme-alpha}. We see that both for $n = 3$ and $n = 6$, it always maximizes expected utility to ask no questions and just guess on all bits.
\end{proof}

\subsection{Proof of Lemma \ref{lem:det}}
\det*
\begin{proof}
    Assume the LLM returns a randomized response. First, utilizing the definition of welfare given in Equation \ref{eq:full-welfare} and the derivation of this definition in the previous section, recall
    \begin{align*}
        G_{0, \gamma}(z, p) &=
        \sum_{\substack{j \leq z \\ \ell \leq n - k - z}} \binom{z}{j}(1-p)^{z - j}p^j \cdot \binom{n - k - z}{\ell} (1-p)^\ell p^{n - k - z - \ell} \\
        &\cdot (n - k - j - \ell)^{1 - \gamma}.
    \end{align*}
    This is the expected welfare for a cluster $0$ member if we choose $z$ of the remaining $n-k$ to be $0$s.
    From this definition, the deterministic welfare was given by
    \begin{align*}
        W_0(k, \gamma)
        &= \sum_{a = 0}^k \binom{k}{a}(1-p)^{k-a}p^a
        \sum_{z=0}^{n-k}
        1[f_\gamma(a) = z]
        G_{0,\gamma}(z, p).
    \end{align*}
    To adapt welfare to randomized policies, define $q_\gamma(a, z) = \Pr[f_\gamma(a) = z]$. In a deterministic policy, we would constrain $q$'s to be either 0 or 1. For randomized policies, welfare for users from cluster $0$ becomes
    \begin{align*}
        W_0(k, \gamma)
        &= \sum_{a = 0}^k \binom{k}{a}(1-p)^{k-a}p^a
        \sum_{z=0}^{n-k}
        q_z(a, \gamma)
        G_{0,\gamma}(z).
    \end{align*}
    Now, let $z^\star$ denote the number of $0$s that is optimal to guess for the remaining bits given $z \in [0, n - k]$. Formally, choose $z^\star = \argmax_{0 \leq z \leq n - k} G_{0, \gamma}(z, p)$.
    Then for users from cluster $0$, the policy that deterministically returns $z^\star$ $0$s has welfare at least as large as the original randomized distribution over $q$s and strictly larger in every situation except for the edge case where whenever the distribution of $q$s has support only over choices of $z$ with $G_{0, \gamma}(z, p) = G_{0, \gamma}(z^\star, p)$. The result for cluster $1$ follows by a symmetrical argument with $G_{1,\gamma}(z, p)$.

    
\end{proof}

\subsection{Proof of Lemma \ref{lem:opt-policy}}
\optimalpolicy*
\begin{proof}
    Consider the plot of optimal response policy as a function of $\gamma$ given in Figure \ref{fig:optimal-policy}. For $\gamma = 10$ in the row second from the left, we can see in the case where $n = 6$ and $k = 3$, it is best to output some number at least a single $0$ \emph{and} a single $1$ when both a $1$ and $0$ are observed.
\end{proof}

\subsection{Proof of Theorem \ref{thm:maximin-policy}}
\maximinpolicy*
\begin{proof}
    To prove that we always utilize the information learned from the revealed features as $\gamma \rightarrow \infty$ and $k \leq n/2$, we solve for the precise optimal response policy and show that it depends on $a$, the number of $1$s revealed. Without loss of generality, shift utility up by $1$ to $(n - k - j - \ell + 1)$.
    Now, consider what happens to the expression $(n - k - j - \ell + 1)^{1 - \gamma}$ as $\gamma$ increases. As $n - k - j - \ell > 1$, the utility will go to $0$. However, when $n - k = j + \ell$, then utility will be exactly 1. Thus, in the limit as $\gamma \rightarrow \infty$, only the terms in the double sum where $j + \ell = n - k$ survive. Because by definition we know $j + \ell$ can never exceed $n - k$, we do not need to worry about negative terms in our expression for welfare. 
    
    The only term where $n - k = j + \ell$ is exactly the case where $j = f_\gamma(a)$ and $\ell = n - k - f_\gamma(a)$. Simplifying our double sum to this single term, we can express welfare as
    \begin{align*}
        W_0(k, \infty) &= \frac{1}{1 - \gamma}\sum_{a = 0}^k \binom{k}{a}(1-p)^{k-a}p^a \\
        &\cdot \left( \binom{f_\gamma(a)}{j}(1-p)^{f_\gamma(a) - j}p^j \cdot \binom{n - k - f_\gamma(a)}{\ell} (1-p)^\ell p^{n - k - f_\gamma(a) - \ell}\right), \\
        W_1(k, \infty) &= \frac{1}{1 - \gamma}\sum_{a = 0}^k \binom{k}{a}(1-p)^ap^{k-a} \\
        &\cdot \left( \binom{f_\gamma(a)}{j}(1-p)^j p^{f_\gamma(a) - j} \cdot \binom{n - k - f_\gamma(a)}{\ell} (1-p)^{n - k - f_\gamma(a) - \ell}p^\ell \right).
    \end{align*}
    This can be further simplified to
    \begin{align*}
        W_0(k, \infty) &= \frac{1}{1 - \gamma}\sum_{a = 0}^k \binom{k}{a}(1-p)^{k-a}p^a \left( p^{f_\gamma(a)} \cdot (1-p)^{n - k - f_\gamma(a)} \right), \\
        W_1(k, \infty) &= \frac{1}{1 - \gamma}\sum_{a = 0}^k \binom{k}{a}(1-p)^ap^{k-a} \left( (1-p)^{f_\gamma(a)} \cdot p^{n - k - f_\gamma(a)} \right).
    \end{align*}
    Now, note that as $\gamma \rightarrow \infty$, then $\frac{1}{1 - \gamma}$ approaches $0$ from the negative side. So if we just want to maximize welfare, we can replace $\frac{1}{1 - \gamma}$ with $-1$. With this and combining $W_0$ and $W_1$, we can write the total expression we wish to maximize as
    \begin{equation*}
        W(k, \infty) = - \sum_{a = 0}^k \binom{k}{a} \left( \alpha \cdot p^{f_\gamma(a)} \cdot (1-p)^{n - k - f_\gamma(a)} + (1 - \alpha) \cdot (1-p)^{f_\gamma(a)} \cdot p^{n - k - f_\gamma(a)} \right).
    \end{equation*}
    This expression will be maximized by choosing a policy to minimize each term within the summation. As can be seen by the dependence of $a$ in the exponent of the terms, $a$ influences the choice of policy. To minimize each term, let $r := \frac{p}{1-p}$ and let us express
    \begin{align*}
        g(f_\gamma(a)) &= \alpha \cdot p^{f_\gamma(a)} \cdot (1-p)^{n - k - f_\gamma(a)} + (1 - \alpha) \cdot (1-p)^{f_\gamma(a)} \cdot p^{n - k - f_\gamma(a)} \\
        &= \alpha \cdot (1 - p)^n \cdot r^{f_\gamma(a) + a} + (1 - \alpha) \cdot p^n \cdot r^{-f_\gamma(a) - a} \\
        &= (1-p)^n \left( \alpha \cdot r^{f_\gamma(a) + a} + (1-\alpha) \cdot r^{n - f_\gamma(a) - a} \right).
    \end{align*}
    To minimize this expression, we can take the derivative and set it to $0$. Ignoring any part that does not depend on $f_\gamma(a)$ gives us a derivative of
    \begin{align*}
        g'(f_\gamma(a)) &= \alpha \cdot \log(r) \cdot r^{f_\gamma(a) + a} + (1 - \alpha) \cdot (-1) \cdot \log(r) \cdot r^{n - f_\gamma(a) - a} \\
        &= \log(r) \cdot (\alpha \cdot r^{f_\gamma(a) + a} - (1 - \alpha) \cdot r^{n - f_\gamma(a) - a}).
    \end{align*}
    Setting this equal to $0$, we get
    \begin{equation*}
        \log(r) \cdot (\alpha \cdot r^{f_\gamma(a) + a} - (1 - \alpha) \cdot r^{n - f_\gamma(a) - a}) = 0.
    \end{equation*}
    If $p = 1/2$, $r = 1$ and $\log(r) = 0$, then all points are stationary points. However, when $p \neq 1 - p$, then a single stationary point exists when
    \begin{equation*}
        f_\gamma(a) = \frac{1}{2}\parens*{n - 2a + \log_r\parens*{\frac{1 - \alpha}{\alpha}}}.
    \end{equation*}
    When $\alpha = 1/2$, this implies $f_\gamma(a)$ should be $n/2 - a$. This is always a non-negative value, as we assume $k \leq n/2$. When $\alpha \neq 1/2$, this implies that as $\alpha$ increases, $f_\gamma(a)$ decreases. As the second derivative of $g$ is positive for valid $p$ and $\alpha$ values, we know the inflection points described must be a minimum. If it is not possible to set $f_\gamma(a)$ as described above within the constraint that $f_\gamma(a) \in [0, n - k]$, then $f_\gamma(a)$ should be set as close to as possible to the calculated expression to minimize $g$. 

    To conclude, we have provided an exact formulation for the optimal response policy when $\gamma \rightarrow \infty$. As it depends on $a$ in all cases, we can claim that the optimal response policy for maximin fairness will never ignore observations.
\end{proof}

\subsection{Proof of Lemma \ref{lem:k-decreasing}}
\kdecreasing*
\begin{proof}
    Consider the plot given in Figure \ref{fig:optimal-queries}. We can see that when $n = 3$ and $p = 0.1$, the optimal number of queries to ask decreases from $1$ to $0$ as $\gamma$ increases.
\end{proof}

\subsection{Proof of Lemma \ref{lem:entropy-decrease}}
\entropydecrease*
\begin{proof}
     Consider the plot given in Figure \ref{fig:optimal-queries}. We can see that when $n = 25$ and $p = 0.1$, the optimal number of queries to ask increases from $4$ to $6$ as $\gamma$ increases from $2$ to $5$. However, in this same plot when $k$ increases from $5$ to $10$, the optimal number of queries to ask decreases back to down to $2$. Computing the entropy based on the optimal response policy first when $k = 4$ and $\gamma = 2$, we find that $H = 1.23$. When $k = 6$ and $\gamma = 5$, the optimal response policy gives an entropy of $H = 2.42$. However, when $k = 2$ and $\gamma = 10$, the limited number of possible responses the model can give results in an entropy of $H = 1.50$. From these three values, we can see that the entropy is a non-monotonic function of $\gamma$.
\end{proof}

\section{Additional Plots}
\label{app-plots}
Using the full expression for welfare with variable $\gamma$ given in Equation \ref{eq:full-welfare}, the following plots were generated to demonstrate how welfare is affected as various parameters are changed.

\begin{figure}[h!]
    \centering
    \includegraphics[width=\linewidth]{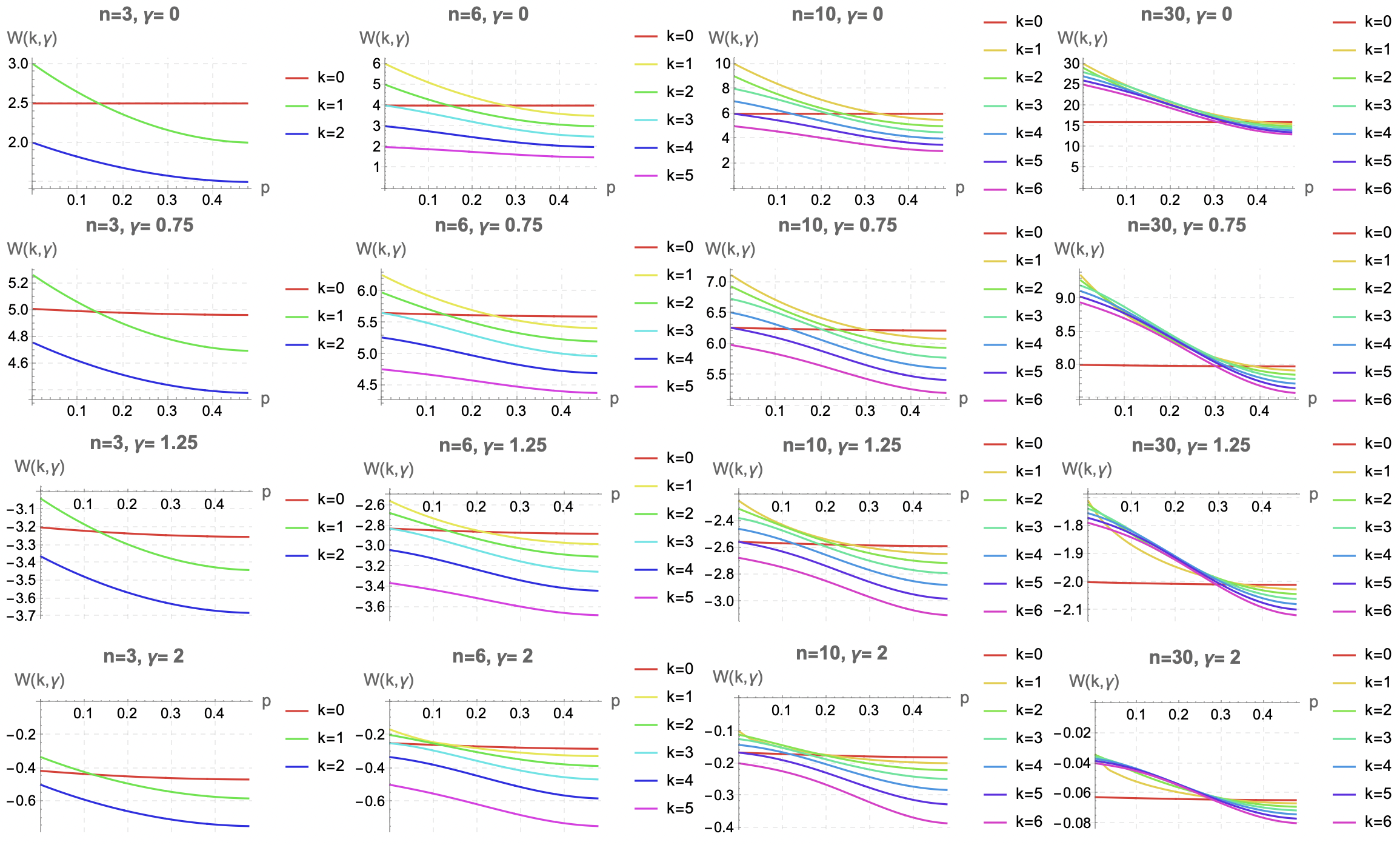}
    \caption{Welfare $W(k, \gamma)$ vs flipping probability $p$ in the case of $\alpha = 1/2$ for various values of $\gamma$. The choice of $k$ leading to the optimal query length is included in each plot.}
    \label{fig:equal-clusters-full}
\end{figure}

\begin{figure}[h!]
    \centering
    \includegraphics[width=\linewidth]{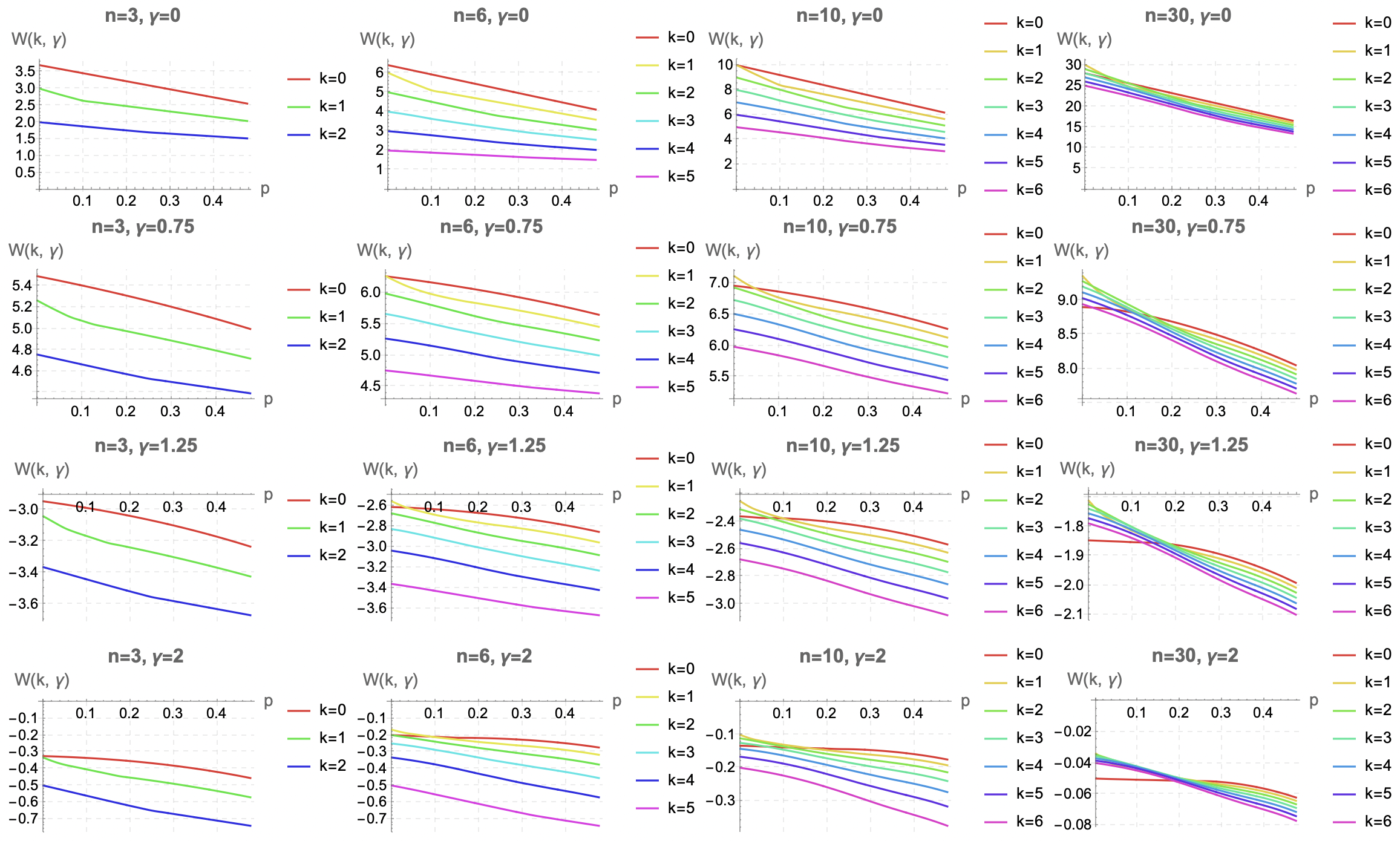}
    \caption{Welfare $W(k, \gamma)$ vs $p$ in the case of $\alpha = 0.9$ for various values of $\gamma$. The choice of $k$ leading to the optimal query length is included in each plot.}
    \label{fig:extreme-alpha-full}
\end{figure}

\begin{figure}[h!]
    \centering
    \includegraphics[width=\linewidth]{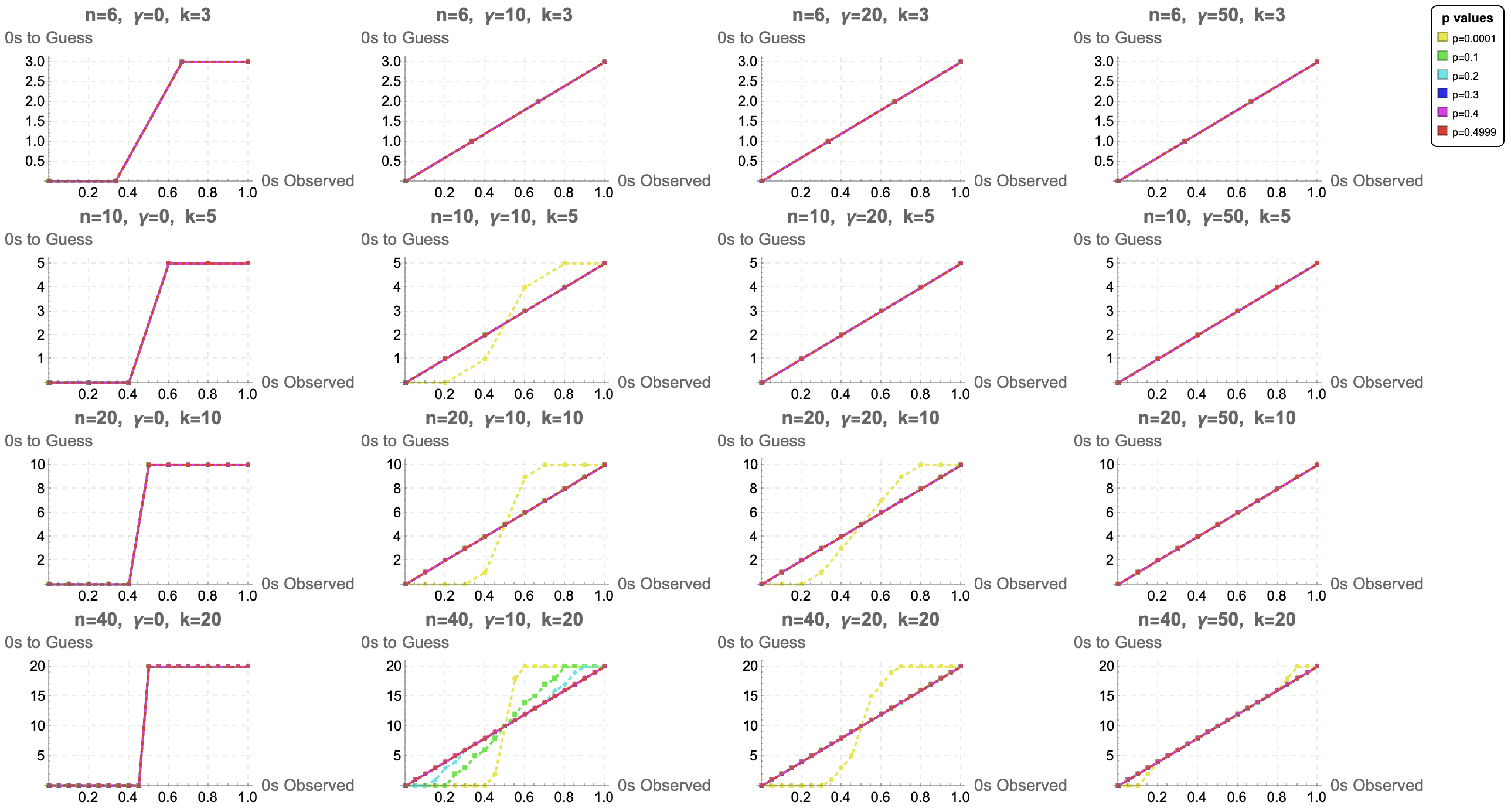}
    \caption{The number of $0$s to output based on the optimal response policy $n - k - f_\gamma(a)$ as a function of $a$ for various $p$ values when $\alpha = 1/2$ and $k = 3$.}
    \label{fig:optimal-policy} 
\end{figure}

\begin{figure}[h!]
    \centering
    \includegraphics[width=\linewidth]{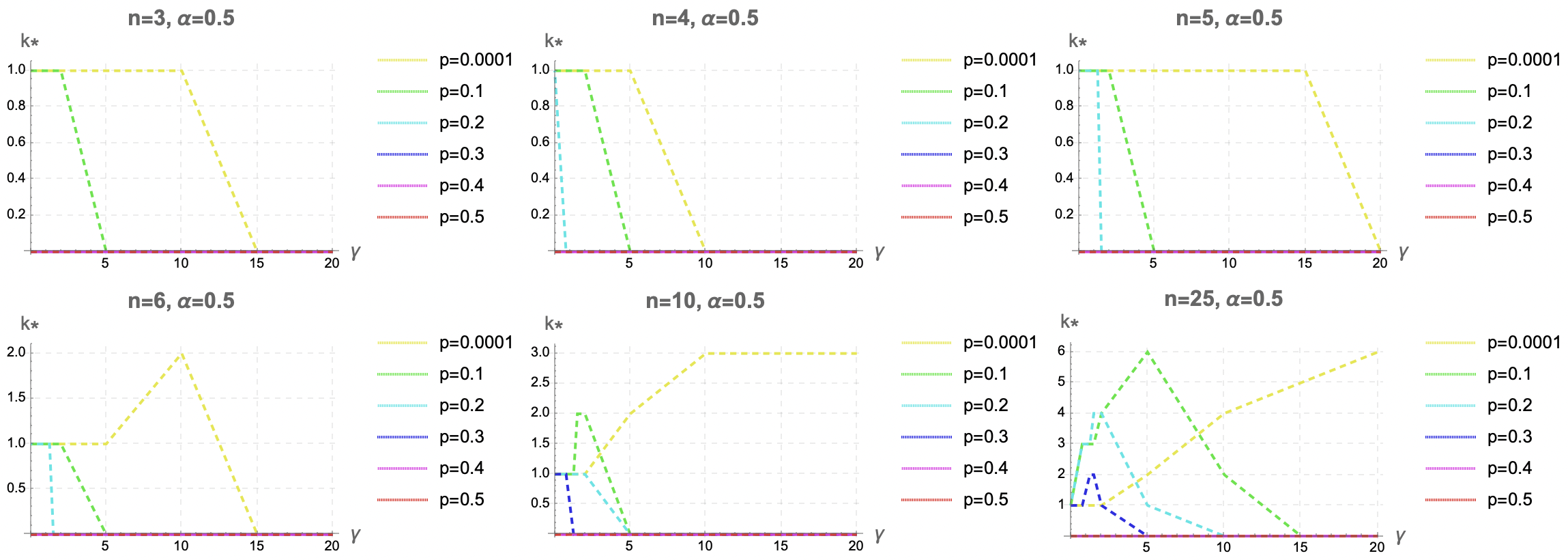}
    \caption{Optimal choice of $k$ vs $\gamma$ for various values of $p$ and $\alpha = 1/2$.}
    \label{fig:optimal-queries}
\end{figure}

\begin{figure}[h!]
    \centering
    \begin{minipage}{0.32\linewidth}
        \centering
        \includegraphics[width=\linewidth]{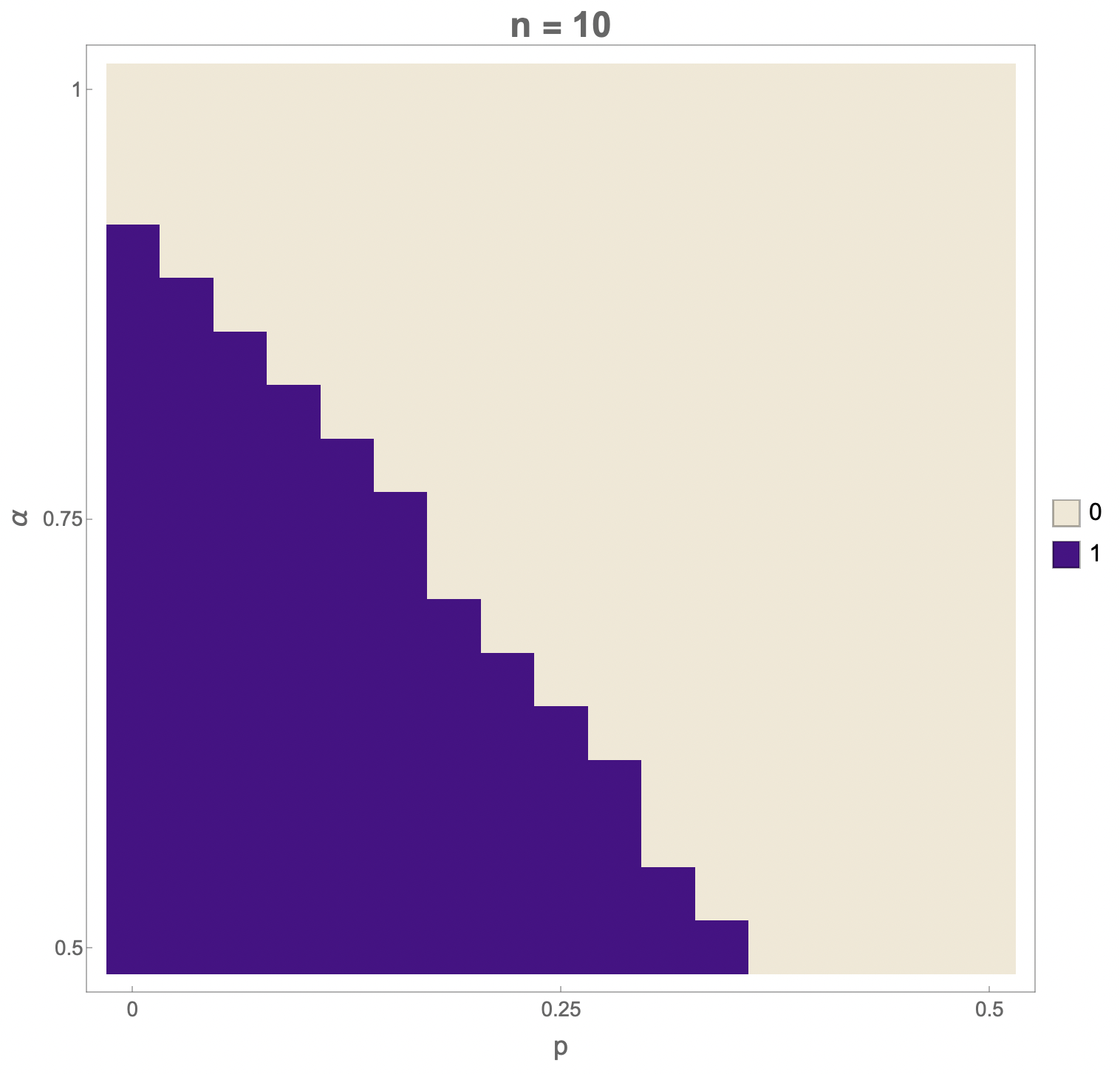}
        \small (a) $n = 10$
    \end{minipage}
    \hfill
    \begin{minipage}{0.32\linewidth}
        \centering
        \includegraphics[width=\linewidth]{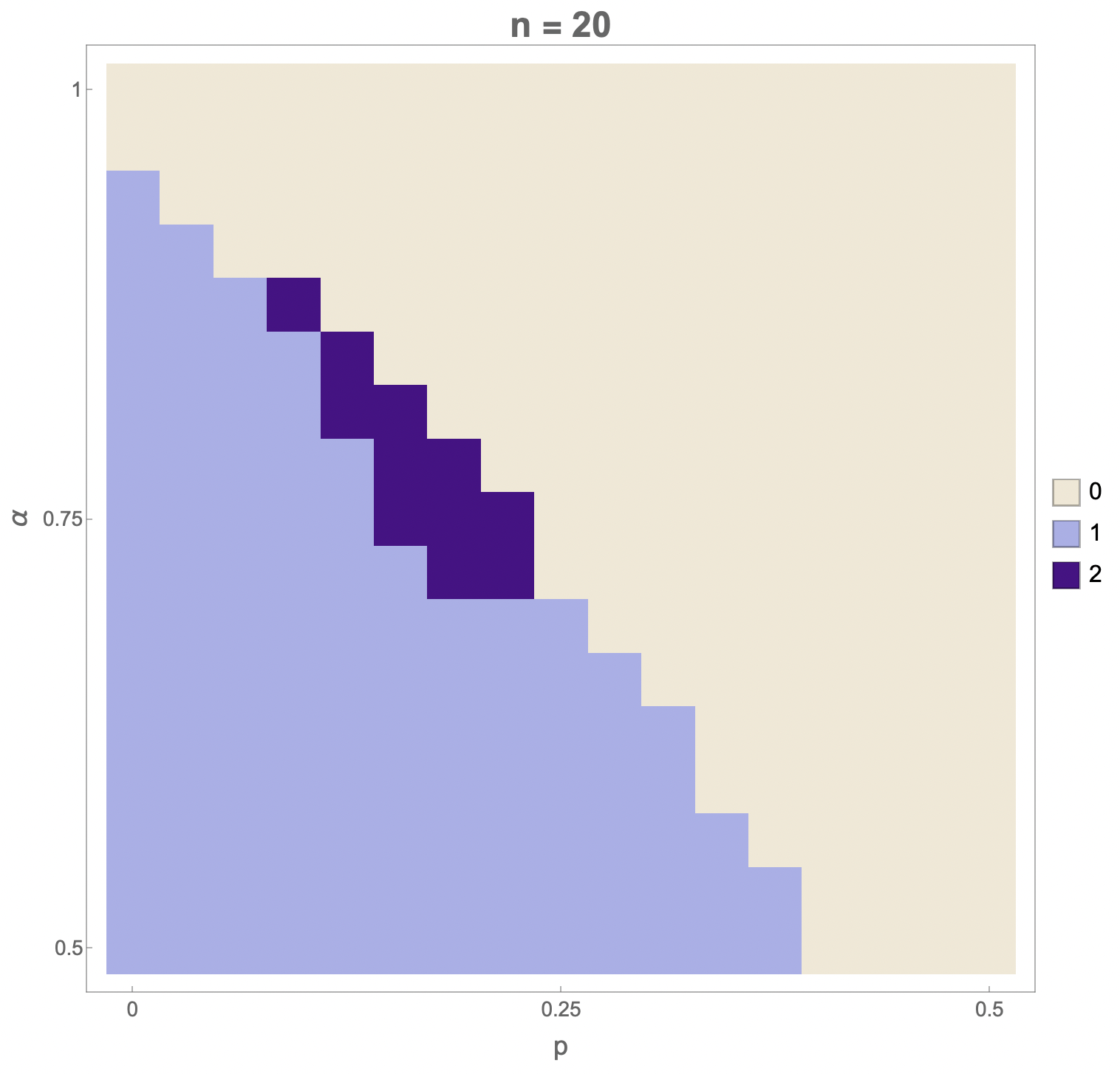}
        \small (b) $n = 20$
    \end{minipage}
    \hfill
    \begin{minipage}{0.32\linewidth}
        \centering
        \includegraphics[width=\linewidth]{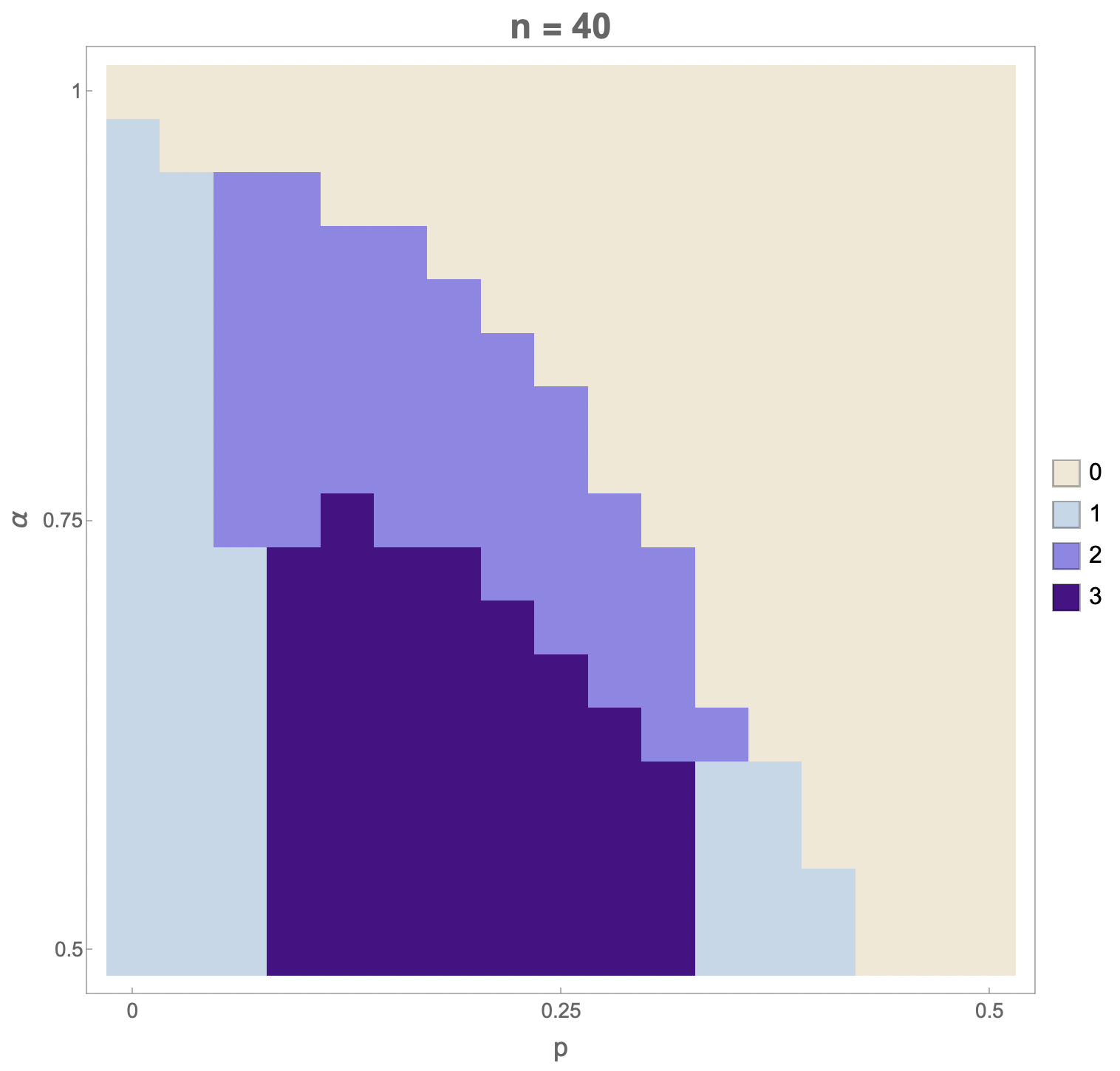}
        \small (c) $n = 40$
    \end{minipage}

    \caption{Optimal number of queries $k$ for various values of $p$ and $\alpha$ across different values of $n$.}
    \label{fig:regime-full}
\end{figure}

\end{document}